\documentclass{LMCS}

\def\dOi{11(2:?72015}
\lmcsheading%
{\dOi}
{1--23}
{}
{}
{Jan.~13, 2013}
{Jun.~11, 2015}
{}

\ACMCCS{[{\bf Theory of computation}]: Logic; Formal languages and
  automata theory; Semantics and reasoning---Program
  semantics\,/\,Program reasoning}


\usepackage{logic9-thm}
\usepackage{tikz}%
  \usepackage[colorinlistoftodos]{todonotes}%

\newcommand{\expand}{\mathit{expand}}
\newcommand{\BT}{BT}
\newcommand{\rBT}{\mathit{rBT}}
\newcommand{\ABT}{ABT}
\newcommand{\promote}{^\uparrow}
\newcommand{\dar}{\!\downarrow}

\newcommand{\order}{\mathit{order}}

\newcommand{\iterate}{\operatorname{\textit{iterate}}}

\newcommand{\val}{\upsilon}

\newcommand{\costep}{\nearrow}%
\newcommand{\ttop}{\top\!\!\!\!\top}%
\newcommand{\bbot}{\bot\!\!\!\!\bot}%
\newcommand{\fix}{\mathrm{fix}}%
\newcommand{\Fix}{\mathrm{Fix}}%
\newcommand{\mon}[1]{\mathrm{mon}\lbrack #1 \rbrack}%
\newcommand{\mbrack}[1]{\left\lbrack #1 \right\rbrack}%

\begin{document}

\title{Using models to model-check recursive schemes}
\author[S.~Salvati]{Sylvain Salvati\rsuper a}	
\address{{\lsuper a}INRIA, LaBRI, 351, cours de la Libération F-33405 Talence France}	
\email{sylvain.salvati@labri.fr}  
\author[I.~Walukiewicz]{Igor  Walukiewicz\rsuper b}
\address{{\lsuper b}CNRS, LaBRI, 351, cours de la Libération F-33405 Talence France}
\email{igw@labri.fr}
\thanks{{\lsuper{a,b}}This work has been
    supported by ANR 2010 BLANC 0202 01 FREC}

  \titlecomment{{\lsuper *}This paper is a long version of a paper that was
    presented at TLCA 2013~\cite{SWTLCA13}}

\begin{abstract}             
  We propose a model-based approach to the model checking problem for
  recursive schemes. Since simply typed lambda calculus with the
  fixpoint operator, $\l Y$-calculus, is equivalent to schemes, we
  propose the use of a model of $\l Y$-calculus to discriminate the terms that
  satisfy a given property. If a model is finite in every type, this
  gives a decision procedure. We provide a construction of such a
  model for every property expressed by automata with trivial
  acceptance conditions and divergence testing. Such properties pose
  already interesting challenges for model construction.  Moreover, we
  argue that having models capturing some class of properties has
  several other virtues in addition to providing decidability of the
  model-checking problem. As an illustration, we show a very simple
  construction transforming a scheme to a scheme reflecting a 
  property captured by a given model.
\end{abstract}

\maketitle

\section{Introduction}
\label{sec:introduction}

We are interested in the relation between the effective
denotational semantics of the simply typed $\lambda Y$-calculus and
the logical properties of B\"ohm trees. By \emph{effective
  denotational} semantics we mean semantic spaces in which the
denotation of a term can be computed; in this paper, these effective
denotational semantics will simply be finite models of the $\l
Y$-calculus, but $Y$ will often be interpreted neither as the
least nor as the greatest fixpoint.

Understanding  properties of B\"ohm trees from a logical point of
view is a problem that arises naturally in the model checking of
higher-order programs.  Often this problem is presented in the
context of higher-order recursive schemes that generate a possibly
infinite tree.  Nevertheless, higher-order recursive schemes can be
represented faithfully by $\l Y$-terms, in the sense that the infinite
trees they generate are precisely the B\"ohm trees $\l Y$-terms define.

The technical question we address here is whether the B\"ohm tree of a given
term is accepted by a given tree automaton. We consider only automata
with trivial acceptance conditions which we call \emph{TAC
  automata}. The principal technical challenge we face is that we allow
automata to detect if a term has a head normal form. We call such
automata \emph{insightful} as opposed to \emph{$\W$-blind} automata that are
insensitive to divergence. For example, the models studied by Aehlig
 or
 Kobayashi~\cite{DBLP:journals/lmcs/Aehlig07,DBLP:conf/popl/Kobayashi09}
 are $\W$-blind. 
The construction of a model of the $\l Y$-calculus that can at the
same time represent safety properties (as defined by trivial automata)
and check whether a computation is diverging is truly challenging.
Indeed, non-convergence has to have a
non-standard interpretation, and this affects strongly the way the
interpretations of terms are computed.
As we show here,  $Y$ combinators cannot be interpreted as an extremal
fixpoint in this case, so 
known algorithms for verification of safety properties cannot take
non-convergence into account in a non-trivial way.

Let us explain the difference between insightful and $\W$-blind
conditions. The definition of a B\"ohm tree says that if the head
reduction of a term does not terminate then in the resulting tree we
get a special symbol $\W$. Yet this is not how this issue is treated
in all known solutions to the model-checking problem. There, instead
of reading $\W$, the automaton is allowed to run on the infinite sequence
of unproductive reductions. In the case of automata with trivial
conditions, this has as an immediate consequence that such an infinite
computation is accepted by the automaton.  From a denotational
semantics perspective, this amounts to interpreting the fixpoint
combinator $Y$ as a greatest fixpoint on some finite monotonous model.
So, for example, with this approach to semantics, the language of
schemes that produce at least one head symbol is not definable by
automata with trivial conditions.  Let us note that this problem
disappears once we consider B\"uchi conditions as they permit one to
detect an infinite unproductive execution. So here we look at a
particular class of properties expressible by B\"uchi conditions.
In summary, the problem we address is a non-trivial extension of what is
usually understood as verification of safety properties for recursive schemes. 

Our starting point is the proof that the usual methods for treating
the safety properties of higher-order schemes cannot capture the
properties described with insightful automata.  
The first result of the paper  shows that extremal fixpoint
models can only capture boolean combinations of $\W$-blind TAC automata.
Our main result is the construction of a model capturing insightful automata.  
This construction
is based on an interpretation of the fixpoint operator
which is neither the greatest nor the least one.  The main difficulty is
to obtain a definition that guaranties the existence and uniqueness of
the fixpoint at every type.

In our opinion, providing models capturing certain classes of
properties is an important problem both from foundational and
practical points of view.  On the theoretical side, models
need to handle all the constructions of the $\l$-calculus while, for
example, the type systems proposed so far by
Kobayashi~\cite{DBLP:conf/popl/Kobayashi09}, and by Kobayashi and
Ong~\cite{kobayashi09:_type_system_equiv_to_modal} do not cater for
$\lambda$-abstraction. Moreover, in op.\ cit.\ the treatment of recursion is
performed by means of a parity game that is not incorporated with the
type system. In contrast, we interpret the $Y$ combinator as an element
of the model we construct.
On the practical side,
models capturing classes of properties set the stage to define
algorithms to decide these properties in terms of evaluating
$\l$-terms in them.  One can remark that models offer most of the
algorithmic advantages of other approaches. As illustrated
by~\cite{salvati12:_loader_urzyc_logic_relat},  the
typing discipline of~\cite{DBLP:conf/popl/Kobayashi09} can be
completely rephrased in terms of simple models.   
More generally, model theoretic methods based on duality offer ways 
to transform questions about the value of $\l Y$-terms in models into
typing problems. Such methods have been largely explored
in~\cite{abramsky91:_domain}.
This approach should allow one to
transfer the algorithms based on types to the approach based on
models. 
This practical
interest of models has been made into a slogan by
Terui~\cite{terui12:rta}: \emph{better semantics, faster computation}.
To substantiate further the interest of models we also present a
straightforward transformation of a scheme to a scheme reflecting a
given property~\cite{Broadbent:2010:RSL:1906484.1906730}.  From a
wider perspective, the model based approach opens a new bridge
between the $\l$-calculus and model-checking communities. In particular,
the model we construct for insightful automata brings into the front
stage particular non-extremal fixpoints. To our knowledge these have
not been studied much in the $\l$-calculus literature.
\medskip

\noindent{\bf Related work} The model checking problem has been solved
by Ong~\cite{DBLP:conf/lics/Ong06} and subsequently revisited in a
number of
ways~\cite{Hague08:collapsible_pushdown_automata_and_recursion_schemes,kobayashi09:_type_system_equiv_to_modal,salvati_walukiekicz11:_krivin_machin_higher_order_schem}.
A much simpler proof for the same problem in the case of $\W$-blind
TAC automata has been given by
Aehlig~\cite{DBLP:journals/lmcs/Aehlig07}. In his influential work,
Kobayashi~\cite{DBLP:conf/popl/Kobayashi09,DBLP:conf/asian/Kobayashi09,DBLP:conf/aplas/Kobayashi09}
has shown that many interesting properties of higher-order recursive
programs can be analyzed with recursive schemes and $\W$-blind TAC
automata.  He has also proposed an intersection type system for the
model-checking problem.  The method has been applied to the
verification of higher-order
programs~\cite{DBLP:conf/fossacs/Kobayashi11}. Another method based on
higher-order collapsible pushdown automata uses invariants expressed
in terms of regular properties of higher-order stacks that is close in
spirit to intersection types \cite{BCHS12}.  Let us note that at
present all algorithmic effort concentrates on $\W$-blind TAC
automata.  Ong and Tsukada~\cite{OngT12} provide a
game semantics model corresponding to Kobayashi's style of type
system. Their model can handle only $\W$-blind automata, but then,
thanks to game semantics, it is fully abstract. In recent
work~\cite{tsukada2014compositional} they extend this method to all
parity automata. The obtained model is infinitary though. 
 We cannot hope to have the
full abstraction in our approach using simple constructions; moreover
it is well-known that it is in general not possible to effectively
construct fully abstract models even in the finite
case~\cite{DBLP:journals/tcs/Loader01}.  In turn, as we mention
in~\cite{WalMFCS12} and show here, handling $\W$-blind automata with
simple models is straightforward. The reflection property for schemes
has been proved by Broadbent
et.\ al.~\cite{Broadbent:2010:RSL:1906484.1906730}. Haddad gives a
direct transformation of a scheme to an equivalent scheme without
divergent computations~\cite{DBLP:journals/corr/abs-1202-3498}.
\medskip

\noindent{\bf Organization of the paper} The
next section introduces the objects of our study: $\l Y$-calculus and
automata with trivial acceptance conditions (TAC automata). 
In Section~\ref{sec:prop-capt-great} we present the correspondence between models
of $\l Y$ with greatest fixpoints and boolean combinations of
$\W$-blind TAC automata. In Section~\ref{sec:model-capt-conv} we give
the construction of the model for insightful TAC automata. The last
section presents a transformation of a term into a term reflecting a
given property.



\section{Preliminaries}\label{sec:preliminaries}

The two basic objects of our study are: $\l Y$-calculus and TAC
automata. We will look at $\l Y$-terms as mechanisms for generating
infinite trees that are then accepted or rejected by a TAC
automaton. The definitions we adopt are standard ones in the
$\l$-calculus and in the automata theory. The only exceptions are the notion
of a tree signature used to simplify the presentation, and the notion
of $\W$-blind/insightful automata that are specific to this paper.

\subsection{\texorpdfstring{$\lambda Y$}{lambda Y}-calculus and models}\label{sec:lambda-y-calculus}
The \emph{set of types} $\Tt$ is constructed from a unique \emph{basic type}
$0$ using a binary operation $\to$. Thus $0$ is a type and if $\a$,
$\b$ are types, so is $(\a\to \b)$. The order of a type is defined by:
$\order(0)=0$, and $\order(\a\to\b)=max(1+\order(\a),\order(\b))$.  We
assume that the symbol $\to$ associates to the right. More
specifically we shall write $\a_1\to \dots \to \a_n \to \b$ so as to
denote the type $(\a_1\to (\dots(\a_{n-1} \to (\a_n\to \b))\dots ))$.

A \emph{signature}, denoted $\S$, is a set of typed constants, \textit{i.e.}
symbols with associated types from $\Tt$. We will assume that for
every type $\a\in \Tt$ there are constants $\w^\a$, $\W^\a$ and
$Y^{(\a\to\a)\to\a}$. A constant $Y^{(\a\to\a)\to\a}$ will stand for a
fixpoint operator. Both $\w^\a$ and $\W^\a$ will stand for undefined terms.
The reason why we need two different constants to denote undefined
terms is clarified in
Section~\ref{sec:model-capt-conv}.  

Of special interest to us will be
\emph{tree signatures} where all constants other than $Y$, $\w$ and
$\W$ have order at most $1$. Observe that types of order $1$ have the
form $0^i\to 0$ for some $i$; the latter is a short notation for $0\to
0\to\dots\to 0\to 0$, where there are $i+1$ occurrences of $0$.

\textbf{Proviso:}\label{proviso} to simplify the notation we will
suppose that all the constants in a tree signature are either of type
$0$ or of type $0\to 0\to 0$. So they are either a constant of the
base type or a function of two arguments over the base
type. This assumption does not affect the results of the paper. \medskip

The set of \emph{simply typed $\l$-terms} is defined inductively as
follows. A constant of type $\a$ is a term of type $\a$. For each type
$\a$ there is a countable set of variables $x^\a,y^\a,\dots$ that are
also terms of type $\a$. If $M$ is a term of type $\b$ and $x^\a$ a
variable of type $\a$ then $\l x^{\a}.M$ is a term of type
$\a\to\b$. Finally, if $M$ is of type $\a\to\b$ and $N$ is a term of
type $\a$ then $MN$ is a term of type $\b$. We shall use the usual
convention about dropping parentheses in writing $\l$-terms and we
shall write sequences of $\l$-abstractions $\l x_1.\dots \l x_n.M$ with
only one $\l$: $\l x_1\dots x_n.M$. Even shorter, we shall write $\l
\vec{x}. M$ when $\vec{x}$ stands for a sequence of variables.

The usual operational semantics of the $\lambda$-calculus is given by
$\beta$-contraction. To give the meaning to fixpoint constants we use
$\delta$-contraction ($\to_\d$).  Of course those rules may be applied
at any position in a term:
\begin{equation*}
(\l x. M)N\to_\b M[N/x]\qquad YM\to_\d M (YM).  
\end{equation*}
We write $\to^*_{\b\d}$ for the $\b\d$-reduction, the reflexive and
transitive closure of the sum of the two relations (we write
$\to^+_{\b\d}$ for its transitive closure). This relation defines an
operational equality on terms. We write $=_{\b\d}$ for the smallest
equivalence relation containing $\to^*_{\b\d}$. It is called
\emph{$\b\d$-conversion} or \emph{$\b\d$-equality}. Given a term $M =
\l x_1\dots x_n.N_0N_1\dots N_p$ where $N_0$ is of the form $(\l
x.P)Q$ or $YP$, then $N_0$ is called the \emph{head redex} of $M$.  We
write $M\to_{ h}M'$ when $M'$ is obtained by $\b\d$-contracting the
head redex of $M$ (when it has one). We write $\to_{ h}^*$ and
$\to_{ h}^+$ respectively for the reflexive and transitive closure
and the transitive closure of $\to_{ h}$.  The relation $\to_{
  h}^*$ is called \emph{head reduction}. A term with no head redex is
said to be in \emph{head normal form}.

Thus, the operational semantics of the $\lambda Y$-calculus is the
$\beta\delta$-reduction. It is well-known that this semantics is
confluent~\cite{statman04} and enjoys subject reduction (\textit{i.e.} the type of
terms is invariant under $\b\d$-reduction). So every term has at most one
normal form, but due to $\delta$-reduction there are terms without a
normal form.  A term may not have a normal form because it does not
have head normal form, in such case it is called
\emph{unsolvable}. Even if a term has a head normal form, i.e.\ it is
\emph{solvable}, it may contain an unsolvable
subterm that prevents it from having a normal form. Finally, it may be
also the case that all the subterms of a term are solvable but the
reduction generates an infinitely growing term. It is thus classical
in the $\lambda$-calculus to consider a kind of infinite normal form
that by itself is an infinite tree, and in consequence it is not a
term of the $\lambda
Y$-calculus~\cite{Barendregt84,amadio98:_domain_lambd_calcul}. This infinite
normal form is called a \emph{B\"ohm tree}.

A \emph{B\"ohm tree} is an unranked, ordered, and potentially infinite
tree with nodes labeled by terms of the form $\l x_1.\dots x_n. N$;
where $N$ is a variable or a constant and $n\geq 0$ (so, in particular,
the sequence of $\lambda$-abstractions may be empty). So for example
$x^0$, $\W^0$, $\l x^0. \w^0$ are labels, but $\l y^0.\ x^{0\to 0}
y^0$ is not.

\begin{defi}\label{df:Bohm tree}
A \emph{B\"ohm tree} of a term $M$ is obtained in the following way. 
\begin{itemize}
\item If $M\to^*_{\b\d} \l \vec x.N_0N_1\dots N_k$ with $N_0$ a
  variable or a constant then $BT(M)$ is a tree having root 
  labeled  by $\l \vec x.N_0$ and having
  $BT(N_1)$, \dots, $BT(N_k)$ as subtrees.
\item Otherwise $BT(M)=\W^\a$, where $\a$ is the type of $M$.
\end{itemize}
\end{defi}
Observe that a term $M$ without the constants $\W$ and $\w$ has a
$\b\d$-normal form if and only if $BT(M)$ is a finite tree without the
constants $\W$ and $\w$. In this case the B\"ohm tree is just another
representation of the normal form. Unlike in the standard theory of
the simply typed $\l$-calculus we will be rather interested in terms
with infinite B\"ohm trees.

Recall that in a tree signature all constants except  $Y$, $\W$, and
$\w$ are of type $0$ or $0\to 0\to 0$. A closed term  without
$\l$-abstraction and $Y$ over such a signature is just a finite binary
tree, where constants of type $0$ occur at leaves, and constants of type
$0\to 0\to 0$  are in the internal nodes. The same holds for B\"ohm
trees:

\begin{lem}\label{lemma:Bohm tree is a tree}
  If $M$ is a closed term of type $0$ over a tree signature then
  $BT(M)$ is a potentially infinite binary tree.
\end{lem}

We will consider finitary models of the $\lambda Y$-calculus. In the
first part of the paper we will concentrate on those where $Y$ is
interpreted as the greatest fixpoint.  The models interpreting $Y$ as
least fixpoints are dual and capture the same class of properties as
the models based on greatest fixpoints for interpreting the $Y$
combinator.

\begin{defi}\label{df:GFP model}
  A \emph{GFP-model} of a signature $\S$ is a tuple
  $\Ss=\struct{\set{\Ss_\a}_{\a\in\Tt},\r}$ where $\Ss_0$ is a finite
  lattice, called the \emph{base set} of the model, and for every type $\a\to\b\in \Tt$, $\Ss_{\a\to\b}$ is the
  lattice $\mon{\Ss_\a\to\Ss_\b}$ of monotone functions from $\Ss_\a$
  to $\Ss_\b$ ordered coordinatewise. The valuation function $\r$ is
  required to satisfy certain conditions:
  \begin{itemize}
  \item If $c\in\S$ is a constant of type $\a$ then $\r(c)$ is an
    element of $\Ss_\a$.  
  \item For every $\a\in \Tt$, both $\r(\w^\a)$ and $\r(\W^\a)$ are
    the greatest elements of $\Ss_\a$.
  \item Moreover, $\r(Y^{(\a\to\a)\to\a})$ is the function assigning
    to every function $f\in \Ss_{\a\to\a}$ its greatest fixpoint.
  \end{itemize}

\end{defi}

\noindent Observe that every $\Ss_\a$ is finite and is thus a complete lattice. 
Hence all the greatest fixpoints
exist without any additional assumptions.

A \emph{variable assignment} is a function $\val$ associating to a
variable of type $\a$ an element of $\Ss_\a$. If $s$ is an element of
$\Ss_\a$ and $x^\a$ is a variable of type $\a$ then $\val[s/x^\a]$ denotes
the valuation that assigns $s$ to $x^\a$ and that is identical to $\val$
everywhere else. 

The \emph{interpretation of a term} $M$ of type $\a$ in the model
$\Ss$ under the valuation $\val$ is an element of $\Ss_\a$ denoted
$\sem{M}_\Ss^\val$. The meaning is defined inductively:
\begin{itemize}
\item $\sem{c}^\val_\Ss=\r(c)$
\item $\sem{x^\a}^\val_\Ss=\val(x^\a)$
\item
  $\sem{MN}^\val_\Ss=\sem{M}^\val_\Ss(\sem{N}^\val_\Ss)$
\item $\sem{\l x^\a.M}^\val_\Ss$ is a function mapping an element
  $s\in \Ss_\a$ to $\sem{M}^{\val[s/x^\a]}_\Ss$ that by abuse of
  notation we may write $\l s.\sem{M}^{\val[s/x^\a]}_\Ss$.
\end{itemize}
It is well-known that the interpretations of terms are always monotone
functions. We refer the reader to~\cite{amadio98:_domain_lambd_calcul}
for details.  As usual, we will omit subscripts or superscripts in the
notation of the semantic function if they are clear from the context.

Of course a GFP model is sound with respect to
${\b\d}$-conversion. Hence two $\b\d$-convertible terms have the same
semantics in the model.  For us it is important that a stronger
property holds: if two terms have the same B\"ohm trees then they have
the same semantics in the model. For this we need to formally define
the semantics of a B\"ohm tree.

The semantics of a B\"ohm tree is defined in terms of its truncations.
For every $n\in \Nat$, we denote by $BT(M)\dar_n$\label{def:cut-BT} the finite term that
is the result of replacing in the tree $\BT(M)$ every subtree at depth
$n$ by the constant $\w^\a$ of the appropriate type. Observe that if
$M$ is closed and of type $0$ then $\a$ will always be the base type
$0$.  This is because we work with a tree signature. We define
\begin{equation*}\label{eq:BT-semantics}
  \sem{BT(M)}^\val_\Ss=\Land\set{\sem{BT(M)\dar_n}^\val_\Ss\mid n\in \Nat}.  
\end{equation*}

The above definitions are standard for $\l Y$-calculus, or more
generally for PCF~\cite{amadio98:_domain_lambd_calcul}. In particular
the following proposition, in a more general form, can be found as
Exercise 6.1.8 in op.\ cit\footnote{In this paper we work with models
  built with finite lattices and monotone functions which are a
  particular case of the directed complete partial order and
  continuous functions used
  in~\cite{amadio98:_domain_lambd_calcul}. We also use GFP models
  while \cite{amadio98:_domain_lambd_calcul} uses least fixpoints, but
  the duality between those two classes of models makes the proof of
  the proposition similar in the two cases.}.


\begin{prop}\label{prop:semantics-of-BT}
  If $\Ss$ is a finite GFP-model and $M$ is a closed term then:
  $\sem{M}_{\Ss}=\sem{BT(M)}_{\Ss}$.
\end{prop}
Observe that $\W$ is used to denote divergence and $\w$ is used in the
definition of the truncation $\BT(M)\dar_n$. In GFP-models this is
irrelevant as the two constants are required to have the same meaning. Later we
will consider models that distinguish those two constants.

\subsection{TAC Automata}\label{sec:triv-autom-infin}
Let us fix a tree signature $\S$. Recall that this means that apart
from $\w$, $\W$ and $Y$ all constants have order at most
$1$. According to our proviso from page~\pageref{proviso} all
constants in $\S$ have either type $0$ or type $0\to 0\to 0$. In this
case, as we only consider closed terms of type $0$, by
Lemma~\ref{lemma:Bohm tree is a tree}, B\"ohm trees are potentially
infinite binary trees. Let $\S_0$ be the set of constants of type $0$,
and $\S_2$ the set of constants of type $0\to 0\to 0$.

\begin{defi}\label{df:trivial aut}
  A \emph{finite tree automaton with trivial acceptance condition}
  (TAC automaton) over the signature $\S=\S_0\cup\S_2$ is
\begin{equation*}
  \Aa=\struct{Q,\S,q^0\in Q,\d_0:Q\times (\S_0\cup\set{\W})\to
    \set{\ffalse,\ttrue},\d_2: Q\times \S_2 \to \Pp(Q^2)}
\end{equation*}  
where $Q$ is a finite set of states and $q^0\in Q$ is the initial
state. The transition function of the TAC automaton may be subject to the
additional restriction:
\begin{equation*}
  \text{\textbf{$\W$-blind:}\quad $\d_0(q,\W)=\ttrue$ for all $q\in Q$.}
\end{equation*}
An automaton satisfying this restriction is called
\emph{$\W$-blind}. For clarity, we use the term \emph{insightful} to
refer to automata without this restriction.
\end{defi}
Automata are used to define languages of possibly infinite binary
trees. More specifically, an automaton over $\S$ shall define a set of
$\S$-labelled binary trees. These trees are partial functions
$t:\set{1,2}^*\to\S\cup\set{\W}$ such that their domain is a binary
tree: (i) if $uv$ is in the domain of $t$ then so is $u$, (ii) if $u$
is in the domain of $t$ and $t(u)$ is in $\S_2$ then $u1$ and $u2$ are
in the domain of $t$, (iii) if $u$ is in the domain of $t$ and
$t(u)\in \S_0\cup\set{\W}$ then $u$ is called a \emph{leaf}, and if
$uv$ is in the domain of $t$ then $v$ is the empty string.

A \emph{run of $\Aa$ on $t$} is a mapping $r:\set{1,2}^*\to Q$
with the same domain as $t$ and such that: 
\begin{itemize}
\item $r(\e)=q^0$, here $\e$ is the root of $t$.
\item $(r(u1),r(u2))\in \d_2(t(u),r(u))$ if $u$ is an internal node.
\end{itemize}
A run is \emph{accepting} if $\d_0(r(u),t(u))=\ttrue$ for every leaf
$u$ of $t$ . A tree is \emph{accepted by $\Aa$} if there is an
accepting run on the tree. The \emph{language} of $\Aa$, denoted
$L(\Aa)$, is the set of trees that are accepted by $\Aa$.

Observe that TAC automata have acceptance conditions on leaves,
expressed with $\d_0$, but do not have acceptance conditions on
infinite paths. For example, this implies that every run on an
infinite tree with no leaves is accepting. This does not mean of
course that TAC automata accept all such trees as there may be no run
on a particular tree. Indeed it may be the case that $\d_2(q,c) = \es$
for some pairs $(q,c)$.

As underlined in the introduction, all the previous works on automata
with trivial conditions rely on the $\W$-blind restriction.  Let us give
some examples of properties that can be expressed with insightful
automata but not with $\W$-blind automata.\label{ex:properties}

\begin{itemize}
\item The set of terms not having $\W$ in their B\"ohm tree. To
  recognize this set we take the automaton with a unique state
  $q$. This state has transitions on all the letters from $\S_2$. It
  also can end a run in every constant of type $0$ except for $\W$:
  this means $\d_0(q,\W)=\ffalse$ and $\d_0(q,c)=\ttrue$ for all other
  $c$.

\item The set of terms having a head normal form. We take an automaton
  with two states $q$ and $q_\top$. From $q_\top$ the automaton accepts
  every tree. From $q$ it has transitions to $q_\top$ on all the
  letters from $\S_2$, on letters from $\S_0$ it behaves as the
  automaton above.

\item Building on these two examples one can easily construct an
  automaton for a property like ``every occurrence of $\W$ is preceded
  by a constant $err$''.
\end{itemize}
It is easy to see that none of these languages is recognized by any
$\W$-blind automaton since if such an automaton accepts a tree $t$
then it accepts also every tree obtained by replacing a subtree of $t$
by $\W$.  This observation also allows one to show that those languages
cannot be defined as boolean combinations of $\W$-blind automata.


\section{GFP models and \texorpdfstring{$\W$}{W}-blind TAC automata}\label{sec:prop-capt-great}

In this section we show that the recognizing power of GFP models
coincides with that of boolean combinations of $\W$-blind TAC
automata. For every automaton we will construct a model capable of
discriminating the terms accepted by the automaton. For the opposite
direction, we will use boolean combinations of TAC automata to
capture the recognizing power of the model. We
start with the expected formal definition of a set of $\lambda Y$-terms
recognized by a model.
\begin{defi}
  For a GFP model $\Ss$ over the base set $\Ss_0$. The \emph{language
    recognized} by a subset $F\incl \Ss_0$ is the set of closed
  $\lambda Y$-terms $\set{M \mid \sem{M}_\Ss\in F}$.
\end{defi}

We need to introduce some notations that we shall use in the course of the
proofs.  Given a closed term $M$ of type $0$, the tree $BT(M)$ can be
seen as a binary tree $t:\set{1,2}^*\to \S$.  For every node $v$ in the
domain of $t$, we write $M_v$ for the subtree of $t$ rooted at node
$v$.  The tree $\BT(M)\dar_k$ is a prefix of this tree containing
nodes up to depth $k$, denote it $t_k$ (c.f.~definition on page~\pageref{def:cut-BT}). It has three types of leaves:
``cut leaves'' are at depth $k$ and are labelled by $\w$,
``non-converging leaves'' labelled by $\W$, and ``normal leaves''
labelled by a constant of type $0$.  
Every node $v$ in the domain
of $t_k$ corresponds to a subterm  of $\BT(M)\dar_k$ that we
denote $M_{v}^k$.
In particular $M_\e^k$ is $BT(M)\dar_k$ since $\e$ is the root of $BT(M)\dar_k$.

\begin{prop}\label{prop:from_W_blind_TAC_to_GFP_models}
  For every $\W$-blind TAC automaton $\Aa$, the language of $\Aa$ is
  recognized by a GFP model. 
 \end{prop}
\begin{proof}
  For the model $\Ss_\Aa$ in question we take a GFP model with the
  base set 
  $\Ss_0=\Pp(Q)$. This determines $\Ss_\a$ for every type $\a$. It
  remains to define the interpretation of constants other than $\w$, $\W$, or
  $Y$. A constant $c$ of type $0$ is interpreted as a set $\set{q \mid
    \d_0(q,c)=\ttrue}$. A constant $a$ of type $0\to 0\to 0$ is
  interpreted as a function whose value on $(S_0,S_1)\in \Pp(Q)^2$ is
  $\set{q \mid \d_2(q,a)\cap S_0\times S_1\not=\es}$. Finally, for the set
  $F_\Aa$ used to recognize $L(\Aa)$ we will take $\set{S\mid q^0\in
    S}$; recall that $q^0$ is the initial state of $\Aa$.  We want to
  show that for every closed term $M$ of type $0$:
  \begin{equation*}
    BT(M)\in L(\Aa)\quad \text{iff}\quad \sem{M}\in F_\Aa.
  \end{equation*}

  For the direction from left to right, we take a $\l Y$-term $M$ such
  that $BT(M)\in L(\Aa)$, and
  show that $q^0\in\sem{BT(M)}$. This will do as $\sem{BT(M)}=\sem{M}$
  by Proposition~\ref{prop:semantics-of-BT}. Recall that
  $\sem{BT(M)}=\Land\set{\sem{BT(M)\dar_k}\mid k=1,2,\dots}$. So it is
  enough to show that $q^0\in \sem{BT(M)\dar_k}$ for every $k$.


  Let us assume that we have an accepting run $r$ of $\Aa$ on $BT(M)$.
  By induction on the height of $v$ in the domain of $BT(M)\dar_k$ we
  show that $r(v)\in\sem{M^k_v}$. The desired conclusion will follow
  by taking $v=\e$; that is the root of the tree.  If $v$ is a ``cut
  leaf'' then $M^k_v$ is $\w^0$. So $r(v)\in \sem{\w^0}$ since $\sem{\w^0}=Q$. If $v$ is
  a ``non-converging leaf'', then $M^k_v$ is $\W^0$ and $r(v)\in
  Q=\sem{\W^0}$. If $v$ is
  a ``normal'' leaf then $M^k_v$ is a constant $c$ of type $0$. We
  have $r(v)\in\set{q: \d(q,c)=\ttrue}$. If $v$ is an internal node
  then $M^k_v=aM^k_{v1}M^k_{v2}$.  By induction assumption $r(v1)\in
  \sem{M^k_{v1}}$ and $r(v2)\in \sem{M^k_{v2}}$. Hence by definition
  of $\r(a)$ we get
  \begin{equation*}
    r(v)\in
    \sem{M_v}=\r(a)(\sem{M^k_{v1}},\sem{M^k_{v2}})\ .
  \end{equation*}

  For the direction from right to left we take a term $M$ and a state
  $q\in\sem{M}$. We construct a run of $\Aa$ on $\BT(M)$ that starts
  with the state $q$. So we put $r(\e)=q$.  If $M$ has no head normal
  form $BT(M) = \W$ and, using Proposition~\ref{prop:semantics-of-BT},
  the conclusion is immediate as the automaton is $\W$-blind.  If $M$
  has as head normal form a nullary constant $a$, the conclusion
  follows from the definition $\sem{a}$.  Now if $M$ has as head
  normal form $a M_1M_2$, by definition of $\sem{a}$, there is
  $(q_1,q_2)$ in $\d(q,a)$ so that $q_1\in\sem{M_1}$ and
  $q_2\in \sem{M_2}$.We repeat the argument with the state $q_1$ from
  node $1$, and with the state $q_2$ from node $2$.  It is easy to see
  that this gives an accepting run of $\Aa$ on $BT(M)$.\end{proof}



As we are now going to see, the power of GFP models is characterized
by $\W$-blind TAC automata.  We will show that every language
recognized by a GFP model is a boolean combination of languages of
$\W$-blind TAC automata.  For the rest of the subsection we fix a tree
signature $\S$ and a GFP model
$\Ss=\struct{\set{\Ss_\a}_{\a\in\Tt},\r}$ over $\S$.

We construct a family of automata that reflect the model $\Ss$.  We
let $Q$ be equal to the base set $\Ss_0$ of the model. We define $\d_0:Q\times
(\S_0\cup\{\W\}) \to\set{\ffalse, \ttrue}$ and $\d_2: Q\times \S_2 \to
\Pp(Q^2)$ to be the functions such that:
\begin{align*}
\d_0(q,a) =& \ttrue \quad \text{iff\quad $q\leq \r(a)$ \qquad (in
the order of $\Ss_0$)}\\
\d_2(q,a) =& \set{(q_1,q_2)\mid q\leq
  \r(a)(q_1,q_2)}.  
\end{align*}
For $q$ in $Q$, we define $\Aa_q$ to be the automaton with the
starting state $q$ and the other components as above:
\begin{equation*} \Aa_q=\struct{Q,\S,q,\d_0,\d_1}\ .
\end{equation*}  
We have the following lemma:
\begin{lem}\label{lem:GFP_models_and_blind_TAC}
  Given a closed $\l$-term $M$ of type $0$: $BT(M)\in L(\Aa_q)$ iff
  $q\leq \sem{M}$.
\end{lem}

\begin{proof}
  We start by showing that if $\Aa_q$ accepts $BT(M)$ then $q\leq
  \sem{M}$.  Proposition~\ref{prop:semantics-of-BT} reduces this
  implication to proving that $q\leq \sem{BT(M)}$. Since $\sem{BT(M)}
  = \bigwedge\set{\sem{BT(M)\dar_k}\mid k \in \mathbb{N}}$, we need to
  show that for every $k>0$, $q\leq \sem{BT(M)\dar_k}$.  Fix an
  accepting run $r$ of $\Aa_q$ on $BT(M)$.  
  We are going to show that for every $v$ in the domain of
  $BT(M)\dar_k$, $r(v)\leq \sem{M^k_v}$.  This
  will imply that $r(\e) = q\leq \sem{BT(M)}\dar_k$.

  We proceed by induction on the height of $v$.  In case $v$ is a
  ``cut leaf'' (or a ``non-converging'' leaf) then $M^k_v$ is $\w^0$
  (or $\W^0$) and $\sem{M^k_v}$ is the
  greatest element of $\Ss_0$ so that $r(v)$ is indeed smaller than
  $\sem{M^k_v}$.  In case $v$ is a ``normal leaf'' then $M^k_v$ is a
  constant $c$ of type $0$.  Since $r$ is an accepting run, we need to
  have, by definition, $r(v)\leq \r(c) = \sem{M^k_v}$.  In case $v$ is
  an internal node then $M^k_v = aM^k_{v1}M^k_{v2}$, and, by
  induction, we have that $r(vi)\leq \sem{M^k_{vi}}$.  Moreover,
  because $r$ is a run, we need to have $r(v)\leq
  \r(a)(r(v1))(r(v2))$, but since $\r(a)$ is monotone, and $r(vi)\leq
  \sem{M^k_{vi}}$, we have $\r(a)(r(v1))(r(v2)) \leq
  \r(a)(\sem{M^k_{v1}})(\sem{M^k_{v2}}) = \sem{M^k_v}$.  This proves,
  as expected, that $r(v)\leq \sem{M^k_v}$.

  Now given $q\leq \sem{M}$ we are going to construct a run of $\Aa_q$
  on $BT(M)$. Recall that for a node $v$ of $BT(M)$ we use $M_v$ to denote the subtree
  rooted in this node. Take $r$ defined by $r(v)=\sem{M_v}$ for every
  $v$.  We show that $r$ is a run of the automaton
  $\Aa_{\sem{M}}$. Since $q\leq \sem{M}$, by the definitions of $\d_0$
  and $\d_1$, this run can be easily turned into a run of
  $\Aa_q$.

  By definition $r(\e)=\sem{M}=\sem{BT(M)}$.  In
  case $v$ is a leaf $c$, then $r(v) = \r(c)$ and we have
  $\d_0(c,\r(c)) = \ttrue$.  In case $v$ is an internal node labeled
  by $a$, then, by definition
  $\sem{M_v}=\r(a)(\sem{M_{v1}},\sem{M_{v2}})$, so 
  $(\sem{M_{v1}},\sem{M_{v2}})$ is in $\d_1(a,\sem{M_v})$. 
  \end{proof}

  This lemma and Proposition~\ref{prop:from_W_blind_TAC_to_GFP_models}
  allow us to infer the announced correspondence.

\begin{thm}\label{thm:characterisation-for-GFP-models}
  A language $L$ of $\lambda$-terms is recognized by a GFP-model iff
  it is a boolean combination of languages of $\W$-blind TAC automata.
\end{thm}
\begin{proof}
  For the left to right direction take a model $\Ss$ and $p\in
  \Ss_0$. By the above lemma we get that the language recognized by
  $\set{p}$ is
  \begin{equation*}
  L_p = L(\Aa_{p})-\bigcup\set{L(\Aa_q) \mid q\in \Ss_0\land q\neq
    p\land q \leq p}
  \end{equation*}
  So given $F$ included in $\Ss_0$, the language recognized by $F$ is
  $\bigcup_{p\in F} L_p$.

  For the other direction we take an automaton for every basic
  language in a boolean combination. We make a product of the
  corresponding GFP models given by
  Proposition~\ref{prop:from_W_blind_TAC_to_GFP_models}, and take the
  appropriate $F$ defined by the form of the boolean combination of
  the basic languages.
\end{proof}

Using the results in~\cite{salvati12:_loader_urzyc_logic_relat}, it
can be shown that typings in Kobayashi's type
systems~\cite{DBLP:conf/popl/Kobayashi09} give precisely values in GFP
models.


\section{A model for insightful TAC automata}\label{sec:model-capt-conv}

The goal of this section is to present a model capable of recognizing
languages of insightful TAC automata. 
Theorem~\ref{thm:characterisation-for-GFP-models} implies that the
fixpoint operator in such a model can be neither the greatest nor the
least fixpoint. In the first subsection we will construct a model 
that is a kind of composition of a GFP model and a model for detecting
divergence. We cannot just take the product of the two models since we
want the fixpoint computation in the model detecting divergence to
influence the computation in the GFP model. In the second part of this section
we will show how to interpret insightful TAC automata in such a model.

\subsection{Model construction and basic
  properties}\label{sec:model-constr-basic}

We are going to build a model $\Kk$ intended to recognize the
language of a given insightful TAC automaton. This model is built on
top of the standard model $\Dd$ for detecting if a term has a
head-normal form.

The model $\Dd=\struct{\set{\Dd_\a}_{\a\in \Tt},\r}$ is  built from
the two elements lattice $\Dd_0=\set{\bot,\top}$. As
$\Dd_{\a\to\b}$ we take the set of monotone functions from $\Dd_\a$ to
$\Dd_\b$ ordered pointwise. So $\Dd_\a$ is a is finite lattice, for
every type $\a$. We write $\bot_\a$ and $\top_\a$, for the least,
respectively the greatest, element of the lattice $\Dd_\a$.
We interpret 
$\w^\a$ and $\W^\a$ as the least elements of $\Dd_\a$, and $Y^{(\a\to\a)\to\a}$ as
the least fixpoint operator.  So $\Dd$ is a dual of a GFP
model from Definition~\ref{df:GFP model}. The reason for
not taking a GFP model here is that we would prefer to use the greatest
fixpoint later in the construction. To all constants other than $Y$,
$\w$, and $\W$ the interpretation $\r$ assigns the greatest element of the
appropriate type. The following theorem is well-known
(cf~\cite{amadio98:_domain_lambd_calcul} page 130).

\begin{thm}\label{thm:D_and_convergence}
  For every closed term $M$ of type $0$ without $\w$ we have:
  \begin{equation*}
 \text{$BT(M) = \W$\quad iff\quad   $\sem{M}_\Dd = \bot$.    }
  \end{equation*}
\end{thm}

We fix a finite set $Q$ and its subset $Q_\W\incl Q$. Later these will be the set
of states of a TAC automaton, and the set of states from which this
automaton accepts $\W$, respectively.  To capture the power of such an
automaton, we are going to define a model $\Kk(Q,Q_\W)$  of the $\l
Y$-calculus based on an
applicative structure  $\Kk_{Q,Q_\W}= (\Kk_{\a})_{\a\in \Tt}$ and with
a non-standard interpretation of the
fixpoint. Roughly, this model will live inside the product of $\Dd$
and the GFP model $\Ss$ for an $\W$-blind automaton. The idea is that
$\Kk(Q,Q_\W)$ will have a projection on $\Dd$ but not necessarily
on $\Ss$. This allows the model to observe whether a term converges or not,
and at the same time to use this information in computing in the second
component.

\begin{defi}\label{def:K}
  For a given finite set $Q$ and a set $Q_\W\incl Q$, we define a family of
  sets $\Kk_{Q,Q_\W}= (\Kk_{\a})_{\a\in \Tt}$ by
  mutual recursion together with a logical relation $\Ll = (\Ll_{\a})_{\a\in
    \Tt}$ such that $\Ll_\a\subseteq \Kk_\a\times \Dd_\a$:
\begin{enumerate}
\item we  let $\Kk_0 = \{(\top,P) \mid P\subseteq Q\} \cup
  \{(\bot,Q_\W)\}$ with the order: $(d_1,P_1)\leq
  (d_2,P_2)$ iff $d_1\leq d_2$ in $\Dd_0$ and $P_1\subseteq
  P_2$. (cf. Figure~\ref{fig:order})
\item $\Ll_0 = \set{((d,P),d) \mid (d,P)\in \Kk_0}$,
\item $\Kk_{\a\to\b} = \set{f\in \mon{\Kk_\a\to\Kk_\b}\mid \exists_{d\in
  \Dd_{\a\to\b}}.\ \forall_{(g,e)\in \Ll_\a}.\ (f(g),d(e))\in\Ll_\b}$,
\item $\Ll_{\a\to\b} = \set{(f,d)\in\Kk_{\a\to\b}\times\Dd_{\a\to\b}\mid \forall_{(g,e)\in \Ll_\a}.\ (f(g),d(e))\in\Ll_\b}$.
\end{enumerate}

\end{defi}

\begin{figure}[tbhf]
  \centering
  \begin{tikzpicture}[node distance=1.5cm]
    \node (tt) {$(\top,\{1;2\})$};%
    \node[below left of= tt] (bt) {$(\top, \{1\})$};%
    \node[below right of =tt](tb) {$(\top,\{2\})$};%
    \node[below left of =tb] (bb) {$(\top,\es)$};%
    \node[below left of =bt] (w){$(\bot,\{1\})$};%
    \draw (bt)--(tt);%
    \draw (tb)--(tt);%
    \draw (bb)--(bt);%
    \draw (bb)--(tb);%
    \draw (w)--(bt);%
  \end{tikzpicture}
  \caption{The order $\Kk_0$ for $Q = \{1,2\}$ and $Q_\W = \{1\}$}
\label{fig:order}
\end{figure}

\begin{figure}[htbp]
  \centering
  \includegraphics{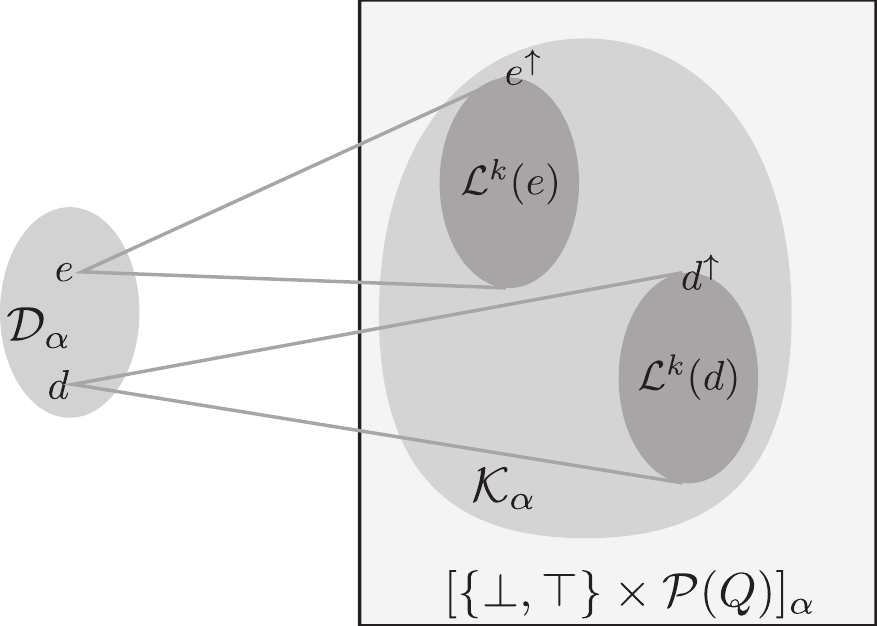}
  \caption{Model $\Dd$ is embeded into model $\Kk$ via logical relation $\Ll$.}
  \label{fig:model}
\end{figure}
Figure~\ref{fig:model} shows the intuition behind the construction.
Every $\Kk_\a$ is finite since it lives inside the standard
model constructed from $\Dd_0\times\Pp(Q)$ as the base set. Moreover,
as we shall see later, for every $\a$, $\Kk_\a$ is a join semilattice
and thus has a greatest element. The logical relation $\Ll$ will divide
$\Kk_\a$ into equivalence classes, one for every element of
$\Dd_\a$. Every equivalence class will also have semilattice structure.

Recall that a TAC automaton is supposed to accept unsolvable terms
from states $Q_\W$. So the unsolvable terms of type $0$ should have
$Q_\W$ as a part of their meaning. This is why $\bot$ of $\Dd_0$ is
associated to $(\bot,Q_\W)$ in $\Kk_0$ via the relation $\Ll_0$. This
also explains why we needed to take the least fixpoint in $\Dd$. If we
had taken the greatest fixpoint then the unsolvable terms would have
evaluated to $\top$ and the solvable ones to $\bot$. In consequence we
would have needed to relate $\top$ with $(\top,Q_\W)$, and we would
have been forced to relate $\bot$ with $(\bot,Q)$. But then
$(\top,Q_\W)$ and $(\bot,Q)$ are incomparable in $\Kk_0$, and this
makes it impossible to construct an order preserving injection from
$\Dd_0$ to $\Kk_0$.

\subsubsection{Structural properties of $\Kk(Q,Q_{\W})$} We are now
going to present some properties of the partial orders $\Kk_\a$.  The
following lemma shows that for every type $\a$, $\Kk_\a$ is a join
semilattice.
\begin{lem}\label{lem:Kk_is_a_join_semilattice}
  Given $(f_1,d_1)$ and $(f_2,d_2)$ in $\Ll_{\a}$, then $f_1\vee f_2$
  is in $\Kk_\a$ and $(f_1\vee f_2,d_1\vee d_2)$ is in $\Ll_{\a}$.
\end{lem}
\begin{proof}
  We proceed by induction on the structure of the type. For the base
  type the lemma is immediate from the definition. For the induction
  step consider a type of a form $\a\to\b$ and assume that $f_1$ and
  $f_2$ in $\mon{\Kk_{\a}\to\Kk_{\b}}$. Since, by induction, $\Kk_\b$
  is a join semilattice, we have that $f_1\vee f_2$ is also in
  $\mon{\Kk_{\a}\to\Kk_{\b}}$. By the assumptions of the lemma, for every
  $(p,e)$ in $\Ll_{\a}$ we have  $(f_1(p),d_1(e))$ and
  $(f_2(p),d_2(e))$ in $\Ll_{\b}$. The induction hypothesis
  implies that $(f_1(p)\vee f_2(p),d_1(e)\vee d_2(e))$ is in
  $\Ll_\b$. As by induction hypothesis $\Kk_{\b}$ is a join
  semilattice, we get $(f_1\vee f_2)(p) = f_1(p)\vee f_2(p)$ is in
  $\Kk_{\b}$.  Thus $((f_1\vee f_2)(p), (d_1\vee d_2)(e))$ is in
  $\Ll_\b$.  Since $(p,e)\in \Ll_\a$ was arbitrary this implies that
  $f_1\vee f_2$ is in $\Kk_{\a\to\b}$ and $(f_1\vee f_2,d_1\vee d_2)$
  is in $\Ll_{\a\to\beta}$.
\end{proof}

A consequence of this lemma and of the finiteness of $\Kk_\a$ is that
$\Kk_\a$ has a greatest element that we denote $\ttop_\a$.  The lemma
also implies the existence of certain meets.

\begin{cor}\label{coro:meets_in_Kk}
  For every type $\a$ and $f_1,f_2$ in $\Kk_\a$. If there is $g\in
  \Kk_\a$ such that $g\leq f_1$ and $g\leq f_2$
  then $f_1$ and $f_2$ have a greatest lower bound $f_1\land
  f_2$. Moreover, if $(f_1,d_1)$ and $(f_2,d_2)$ are in $\Ll_\a$ then
  $(f_1\land f_2,d_1\land d_2)$ is in $\Ll_\a$.
\end{cor}
\begin{proof}
  Let $F = \set{g\in \Kk_\a\mid g\leq f_1 \text{ and } g\leq f_2}$. As
  $\Kk_\a$ is finite, the set $F$ is finite. An iterative use of
  Lemma~\ref{lem:Kk_is_a_join_semilattice} shows that $\bigvee F$
  exists and is in $\Kk_\a$.
  It is then straightforward to see that $\bigvee F$ is
  indeed the greatest lower bound of $f_1$ and $f_2$.

  Now as $\Dd_\a$ is a complete lattice, we also have that $d_1\land
  d_2$ exits.  Then a similar induction as in the proof of
  Lemma~\ref{lem:Kk_is_a_join_semilattice} shows that when $(f_1,d_1)$
  and $(f_2,d_2)$ are in $\Ll_\a$, then $(f_1\land f_2,d_1\land d_2)$
  is in $\Ll_\a$.
\end{proof}

We are now going to show that every constant function of
$\mon{\Kk_\a\to\Kk_\b}$ is actually in $\Kk_{\a\to\b}$.

\begin{lem}\label{lem:constant_functions_are_in_Kk}
  For every  $q$ in $\Kk_\b$, the constant function
  $c_q\in\mon{\Kk_\a\to\Kk_\b}$ assigning $q$ to every element of $\Kk_\a$ is in
  $\Kk_{\a\to\b}$.
\end{lem}

\begin{proof}
  To show that $c_q$ is in $\Kk_{\a\to\b}$, we need to find $h_q$ in
  $\Dd_{\a\to\b}$ such that for every $(p,e)$, $(c_q(p),h_q(e))$ is in
  $\Ll_\b$.  Since $q$ is in $\Kk_\b$, there is $d$ such that $(q,d)$
  is in $\Ll_\b$.  It suffices to take $h_q$ to be the function of
  $\Dd_{\a\to\b}$ such that for every $e$ in $\Dd_{\a}$, $h_q(e) =
  d$.
\end{proof}
As one easily observes that for every $p\in \Kk_\a$,
$\ttop_{\a\to\b}(p) = \ttop_\b$, a consequence of this lemma is that
$(\ttop_\a,\top_\a)$ is in $\Ll_\a$ for every $\a$.

This lemma allows us to define inductively on types the family of
constant functions $(\bbot_\a)_{\a\in \Tt}$ as follows:
\begin{enumerate}
\item $\bbot_0 = (\bot,Q_\W)$,
\item $\bbot_{\a\to\b}(h) = \bbot_{\b}$ for every $h$ in $\Kk_\a$.
\end{enumerate}
Notice that $\bbot_\a$ is a minimal element of $\Kk_\a$, but
$\Kk_\a$ does not have a least element in general.

\subsubsection{Galois connections between $\Kk_\a$ and $\Dd_\a$} In
this part, we wish to show that the relation $\Ll_\a$ is indeed
defining an injection from
$\Kk_\a$ to $\Dd_\a$ that we shall denote with
$(\bar{\cdot})$. Moreover, we are going to define a mapping
$(\cdot)^{\uparrow}$ from $\Dd_\a$ to $\Kk_\a$ so that $(\bar{\cdot})$
and $(\cdot)^{\uparrow}$ define a Galois connection between $\Kk_\a$
and $\Dd_\a$. This Galois connection plays a key role in allowing the
model to track convergence and, thus, in the definition of the
interpretation of fixpoints in the model. We shall also see that both
$(\bar{\cdot})$ and $(\cdot)^\uparrow$ commute with application.

So as to define this Galois connection, we need to introduce the
notion of $\Dd$-completeness of types.  This notion imposes some
basic properties that allow us to construct both $(\bar{\cdot})$ and
$(\cdot)^\uparrow$.  Our goal is to establish that every type is
$\Dd$-complete. 

For every $d$ in $\Dd_\a$, we denote by $L_d$ the set of elements of
$\Kk_\a$ that are related to it:
\begin{equation*}
  L_d=\set{p\in\Kk_\a\mid
    (p,d)\in\Ll_\a}.  
\end{equation*}

\begin{defi}
  A type $\a$  is \emph{$\Dd$-complete} if, for every $d$ in $\Dd_\a$:
    \begin{enumerate}
    \item $L_d$ is not empty,
    \item $\bbot_\a\leq \bigvee L_d$,
    \item for every $(f,e)$ in $\Ll_\a$: $f\leq \bigvee L_d$ iff $e
      \leq d$.
    \end{enumerate}
\end{defi}

Later we will show that every type is $\Dd$-complete, but for this we
will need some preparatory lemmas.

\begin{lem}\label{lem:join_L_d_d_in_Ll}
  If $\a$ is a $\Dd$-complete type and $d$ is in $\Dd_\a$ then $(\bigvee
  L_d,d)$ is in~$\Ll_\a$.
\end{lem}
\begin{proof}
  Since $\a$ is $\Dd$-complete, $L_d$ is not empty, and the conclusion
  follows directly from Lemma~\ref{lem:Kk_is_a_join_semilattice}.
\end{proof}

\begin{lem}\label{lem:embedding_of_Dd_into_Kk_at_Dd_complete_types}
  If $\a$ is a $\Dd$-complete type, and $d,e\in\Dd_\a$ then:
  $e\leq d$ iff $\bigvee L_e\leq \bigvee L_d$.
\end{lem}
\begin{proof}
  As $\a$ is $\Dd$-complete both $L_e$ and $L_d$ are not empty and
  therefore, $\bigvee L_e$ and $\bigvee L_d$ are well-defined.
  Lemma~\ref{lem:join_L_d_d_in_Ll} also gives that $(\bigvee L_e,e)$
  is in $\Ll_{\a}$. Now from $\Dd$-completeness of $\a$, we have that
  $\bigvee L_e\leq \bigvee L_d$ iff $e\leq d$.
\end{proof}

The next step is to define the operation $(\cdot)^\uparrow$ that, as we will show
later, is an embedding of $\Dd$ into $\Kk$.  For this we need the
notion of \emph{co-step functions} that are particular functions from
a partial order $L_1$ to a partial order $L_2$, the latter having the
greatest element $\top_2$. Given two elements $p$ in $L_1$ and $q$ in
$L_2$, the co-step function $p\costep q$ is a function from $\mon{L_1\to L_2}$
such that for $r$ in $L_1$, 
$$(p\costep q)(r) = 
  \begin{cases}
  q & \text{ when }r\leq p\\
  \top_2 & \text{otherwise}\ .
  \end{cases}
 $$
\begin{defi}\label{df:uparrow} Let
  $\a,\b$ be $\Dd$-complete types. For every $h\in \Dd_{\a\to\b}$ and
  every $d\in\Dd_\a$ we define two monotone functions and the element
  $h^\uparrow$:
  \begin{gather*}
    f_{h,d}=\Lor L_d\costep \Lor L_{h(d)},\qquad \bar
    f_{h,d}=d\costep h(d),\\
    h^\uparrow=\Land_{d\in \Dd} f_{h,d}\ .
  \end{gather*}
  For $h$ in $\Dd_0$, we define $h^\uparrow$ to be  $(\bot,Q_\W)$ when $h
  =\bot$, and to be $(\top,Q)$ when $h = \top$.
\end{defi}

The next lemma summarizes all the essential properties of the model
$\Kk$. 
\newcommand{\fhd}{f_{h,d}}
\newcommand{\bfhd}{\bar f_{h,d}}
\begin{lem}\label{lem:uparrow_existence_and_definition}
  For all $\Dd$-complete types $\a$, $\b$, for every $h\in
  \Dd_{\a\to\b}$ and every $d\in\Dd_\a$:
  \begin{enumerate}
  \item $(\fhd,\bfhd)$ is in $\Ll_{\a\to\b}$;
  \item $\bbot_{\a\to\b}\leq \fhd$;
  \item $h^\uparrow$ is an element of $\Kk_{\a\to\b}$ and
    $(h^\uparrow,h)\in \Ll_{\a\to\b}$;
  \item if $(p,e)\in \Ll_\a$ then $h^\uparrow(p)=\Lor L_{h(e)}$;
  \item $h^\uparrow=\Lor L_h$.
  \end{enumerate}
\end{lem}
\begin{proof}
  For the first item we take $(p,e)\in \Ll_\a$, and show that
  $(\fhd(p),\bfhd(e))\in \Ll_\b$. This will be sufficient by the
  definition of $\Ll_{\a\to\b}$.  Lemma~\ref{lem:join_L_d_d_in_Ll}
  gives $(\Lor L_d,d)\in \Ll_\a$ and $(\Lor L_{h(d)},h(d))\in
  \Ll_\b$. By $\Dd$-completeness of $\a$: $p\leq \Lor L_d$ iff $e\leq
  d$. We have two cases. If $p\leq \Lor L_d$ then $\fhd(p)=\Lor
  L_{h(d)}$ and $\bfhd(e)=h(d)$. Otherwise, $p\nleq\Lor L_d$ gives
  $\fhd(p)=\ttop_\b$ and $\bfhd(e)=\top_\b$. With the help of
  Lemma~\ref{lem:join_L_d_d_in_Ll} in both cases we have that the
  result is in $\Ll_\b$, and we are done.

  For the second item, by $\Dd$-completeness of $\b$ we have $\Lor
  L_{h(d)}\geq \bbot_\b$. In the proof of the first item we have seen
  that $\fhd(p)\geq \Lor L_{h(d)}$ for every $p\in \Kk_\a$. Since
  $\bbot_{\a\to\b}(p)=\bbot_\b$ we get $\bbot_{\a\to\b}\leq \fhd$.

  In order to show the third item we use the first item telling us
  that $(f_{h,e},\bar f_{h,e})$ is in $\Ll_{\a\to\b}$ for every $e\in
  \Dd_\a$. Since by the second item $\bbot_\a\leq f_{h,e}$,
  Corollary~\ref{coro:meets_in_Kk} shows that $(\Land_{e\in \Dd_\a}
  f_{h,e},\Land_{e\in \Dd_\a} \bar f_{h,e})$ is in
  $\Ll_{\a\to\b}$. Directly from the definition of co-step functions
  we have $\Land_{e\in \Dd_\a} e\costep h(e)=h$. This gives, as
  desired, $(\Land_{e\in \Dd_\a} f_{h,e},h)$ in $\Ll_{\a\to\b}$.

  For the fourth item, take an arbitrary $(p,e)\in \Ll_\a$. We show
  that $h^\uparrow(p)=\Lor L_{d(e)}$. By definition
  $h^\uparrow(p)=\Land_{e'\in \Dd_\a} f_{h,e'}(p)$. Moreover
  $f_{h,e'}(p)=\Lor L_{h(e')}$ if $p\leq \Lor L_{e'}$, and
  $f_{h,e'}(p)=\ttop_\b$ otherwise. By $\Dd$-completeness of $\a$:
  $p\leq \Lor L_{e'}$ iff $e\leq e'$. So $h^\uparrow(p)=\Land_{e'\in
    \Dd_\a} f_{h,e'}(p)=\Land\set{\Lor L_{h(e')} : e\leq e'}$. By
  Lemma~\ref{lem:embedding_of_Dd_into_Kk_at_Dd_complete_types}, if
  $e\leq e'$ then $\Lor L_{h(e)}\leq \Lor L_{h(e')}$. Hence
  $h^\uparrow(p)=\Lor L_{h(e)}$.

  For the last item we want to show that $h^\uparrow=\Lor L_h$. We
  know that $h^\uparrow\in L_h=\set{g\in \Kk_{\a\to\b} : (g,h)\in
    \Ll_\a}$ since $(h^\uparrow,h)\in \Ll_{\a\to\b}$ by the third
  item. We show that for every $g\in L_h$, $g\leq h^\uparrow$. Take
  some $(p,e)\in \Ll_\a$. We have $(g(p),h(e))\in \Ll_\b$, hence
  $g(p)\leq \Lor L_{h(e)}$ by definition of $L_{h(e)}$.  Since
  $h^\uparrow(p)=\Lor L_{h(e)}$ by the fourth item, we get $g\leq
  h^\uparrow$.
\end{proof}

\begin{lem}\label{lem:every_type_is_Dd_complete}
  Every type $\a$ is $\Dd$-complete.
\end{lem}

\begin{proof}
  This is proved by induction on the structure of the type.  The case
  of the base type follows by direct examination. For the induction
  step consider a type $\a\to\b$ and suppose that $\a$ and $\b$ are
  $\Dd$-complete.  Given $d$ in $\Dd_{\a\to\b}$,
  Lemma~\ref{lem:uparrow_existence_and_definition} gives that
  $(d^\uparrow,d)$ is in $\Ll_{\a\to\b}$ proving that $L_d\neq \es$, it
  also gives that $\bbot_{\a\to\b}\leq d^\uparrow$ and $d^\uparrow=\bigvee
  L_d$, so  we obtain $\bbot_{\a\to\b}\leq \bigvee L_d$.  It just remains
  to prove that for every  $(f,e)$ in $\Ll_{\a\to\b}$: $f\leq \bigvee L_d$ iff
  $e\leq d$.

  We first remark that, as by induction hypothesis, $\a$ and $\b$ are
  $\Dd$-complete, by Lemma~\ref{lem:uparrow_existence_and_definition}
  (items (4) and (5)), for every $(p,e')\in \Ll_\a$ we have:
  \begin{equation}
    \label{eq:1}
    \Lor
    L_{d(e')}=d^\uparrow(p)=\left(\bigvee L_d\right)(p)
  \end{equation}
  
  Let's first suppose that $e\leq d$. Take a $p\in \Kk_\a$. By
  definition of the model there is $e'$, such that $(p,e')\in
  \Ll_\a$. As $\a$ is $\Dd$-complete,
  Lemma~\ref{lem:embedding_of_Dd_into_Kk_at_Dd_complete_types} gives
  us $\bigvee L_{e(e')}\leq \bigvee L_{d(e')}$. By definition of
  $\Ll_{\a\to\b}$ we have that $(f(p),e(e'))\in \Ll_\b$, so $f(p)\leq
  \Lor L_{e(e')}$ by definition of $L_{e(e')}$. This gives $f(p)\leq
  \Lor L_{e(e')}\leq \Lor L_{d(e')}$. Finally Equation~(\ref{eq:1})
  shows the desired $f(p)\leq \left(\bigvee L_d\right)(p)$ for every
  $p\in\Kk_\a$.

  Let us now suppose that $f\leq \bigvee L_d$. The $\Dd$-completeness
  of $\a$ tells us that for every $e'$ in $\Dd_\a$ there is $p$ in $\Kk_{\a}$
  so that $(p, e')$ is in $\Ll_{\a}$.  Then
  Equation~(\ref{eq:1}) gives
  $f(p)\leq \left(\bigvee L_d\right)(p) = \bigvee L_{d(e')}$. Now, as
  by induction $\b$ is $\Dd$-complete, the
  fact that $(f(p),e(e'))\in \Ll_\b$ entails $e(e')\leq d(e')$.  As
  $e'$ was arbitrary we obtain $e\leq d$.
\end{proof}

The proposition below sums up the properties of the embedding
$(\cdot)^\uparrow$ from Definition~\ref{df:uparrow}.

\begin{prop}\label{prop:uparrow_basics}
  Given a type $\alpha$, and $d$ in $\Dd_\a$, the element $d^\uparrow$ from
  $\Kk_\a$ is such that:
  \begin{enumerate}
  \item $(d^\uparrow, d)$ is in $\Ll_\a$,
  \item if $e\in\Dd_\a$ and $d\leq e$ then $d^\uparrow\leq
    e^\uparrow$,
  \item if $(f,d)$ is in $\Ll_\a$, then $f\leq d^\uparrow$,
  \item if $\a = \a_1\to\a_2$ and $(g,e)$ is in $\Ll_{\a_1}$ then
    $d^\uparrow(g) = (d(e))^\uparrow$
  \end{enumerate}
\end{prop}
\begin{proof}
  These properties follow directly from
  Lemma~\ref{lem:uparrow_existence_and_definition}, except for the
  second property for which a small calculation is needed. Since
  $(d^\uparrow,d)$ is in $\Ll_\a$ and $d\leq e$ then by
  Lemma~\ref{lem:uparrow_existence_and_definition}: $d^\uparrow\leq
  \Lor L_e$. The latter is precisely $e^\uparrow$ by
  Lemma~\ref{lem:uparrow_existence_and_definition}.
\end{proof}
In particular, in combination with item 3 of 
Lemma~\ref{lem:uparrow_existence_and_definition} , this
proposition shows that the operator $(\cdot)^\uparrow$ commutes with
the application: $d^\uparrow(e^\uparrow) = (d(e))^\uparrow$.

The next lemma shows that the relation $\Ll_\a$ is functional. 

\begin{lem}\label{lem:unicity_of_left_projection}
  For every type $\a$ and $f$ in $\Kk_\a$: if $(f,d_1)$ and $(f,d_2)$ are
  in $\Ll_{\a}$, then $d_1 = d_2$.
\end{lem}

\begin{proof}
  We proceed by induction on the structure of the type. The case of
  the base type follows from a direct inspection. For the induction
  step suppose that both $(f,d_1)$ and $(f,d_2)$ are in
  $\Ll_{\a\to\b}$. Take an arbitrary $e\in \Dd_\a$. By
  Lemma~\ref{lem:uparrow_existence_and_definition} we have
  $(e^\uparrow,e)\in \Ll_\a$. Therefore $(f(e^\uparrow),d_1(e))$ and
  $(f(e^\uparrow),d_2(e))$ in $\mathcal{L}_{\b}$.  The induction
  hypothesis implies that $d_1(e) = d_2(e)$. Since $e$ was arbitrary
  we get $d_1=d_2$.
\end{proof}

Since, by definition, for every $f\in \Kk_\a$ we have $(f,d)\in\Ll_\a$
for some $d\in\Dd_\a$, the above lemma gives us a  projection
of $\Kk_\a$ to $\Dd_\a$.  For this we re-use the notation we have
introduced in Definition~\ref{df:uparrow}.
\begin{defi}\label{df:bar}
  For every type $\a$ and $f\in \Kk_\a$ we let
  $\bar f$  be the unique element of $\Dd_\a$ such that $(f,\bar
  f)\in \Ll_\a$.
\end{defi}
Notice that $\overline{d^\uparrow} = d$ for every $d$ in $\Dd_\a$,
since $(d^\uparrow,d)$ is in $\Ll_\a$ by Proposition~\ref{prop:uparrow_basics}.

We immediately state some properties of the projection. We start by
showing that it commutes with the application.

\begin{lem}\label{lem:commutation_overline_application}
  Given $f$ in $\Kk_{\a\to\b}$ and $p$ in $\Kk_\a$, $\overline{f(p)} =
  \overline{f}(\overline{p})$.
\end{lem}
\begin{proof}
  We have $(f,\overline{f})$ in $\Ll_{\a\to \b}$ and
  $(p,\overline{p})$ in $\Ll_\a$, so that $(f(p),
  \overline{f}(\overline{p}))$ is in $\Ll_\b$ and thus
  $\overline{f(p)} = \overline{f}(\overline{p})$.
\end{proof}

\begin{lem}\label{lem:overline_ordering_lowering}
  Given $f$ and $g$ in $\Kk_\a$, if $f\leq g$ then
  $\overline{f}\leq\overline{g}$.
\end{lem}

\begin{proof}
  We proceed by induction on the structure of the types. The case of
  the base type follows by a straightforward inspection. For the
  induction step take $f\leq g$ in $\Kk_{\a\to \b}$. For an arbitrary
  $d\in \Dd_\a$ we have $f(d^\uparrow)\leq g(d^\uparrow)$. By
  induction hypothesis on type $\b$ we get $\bar{f(d^\uparrow)}\leq
  \bar{g(d^\uparrow)}$. By
  Lemma~\ref{lem:commutation_overline_application} we obtain
  $\overline{f(d^\uparrow)} = \overline{f}(\overline{d^\uparrow}) =
  \overline{f}(d)$. The last equality follows from the fact that
  $\bar{d^\uparrow}=d$ since $(d^\uparrow,d)$ is in $\Ll_\a$ by
  Proposition~\ref{prop:uparrow_basics}. Of course the same equalities
  hold for $g$ too. So $\bar f(d)\leq \bar g(d)$ for arbitrary $d$,
  and we are done.
\end{proof}
Taking an abstract view on the operations $(\cdot)^\uparrow$ and
$\overline{(\cdot)}$, we can summarise all the properties we have
shown as follows:
\begin{cor}
For the models $\Dd$ and $\Kk$ as defined above.
\begin{enumerate}
\item Mapping $(\cdot)^\uparrow$ is a functor from $\Dd$  to $\Kk$.
\item Mapping $\overline{(\cdot)}$ is a functor from $\Kk$ to $\Dd$.
\item At every type both mappings are monotonous and moreover they
  form a Galois connection in the sense that $\overline{f}\leq d$ iff
  $f\leq d^\uparrow$.
\item The pair $\overline{(\cdot)}$, $(\cdot)^\uparrow$ forms a retraction:
  $\overline{d^\uparrow} = d$.
\end{enumerate}
\end{cor}

\subsubsection{Interpretation of fixpoints} We are now going to give
the definition of the interpretation of the fixpoint combinator in
$\Kk$.  This definition is based on that of the fixpoint operator in
$\Dd$. We write $\fix_\a$ for the operation in
$\Dd_{(\a\to\a)\to\a}$ that maps a function of $\Dd_{\a\to\a}$ to its
least fixpoint.

\begin{lem}\label{lem:fixpoint_definition}
  Given $f$ in $\Kk_{\a\to\a}$, we have
  $f(\fix_\a(\overline{f})^\uparrow)\leq \fix_\a(\overline{f})^\uparrow$.
\end{lem}

\begin{proof}
  By proposition~\ref{prop:uparrow_basics}, 
  $(\fix_\a(\overline{f})^\uparrow,\fix_\a(\overline{f}))$ is in
  $\mathcal{L}_\a$.  Moreover, as $(f,\overline{f})$ is in
  $\Ll_{\a\to\a}$, by definition of $\Ll_{\a\to\a}$, we have
  $(f(\fix_\a(\overline{f})^\uparrow),\overline{f}(\fix_\a(\overline{f})))
  = (f(\fix_\a(\overline{f})^\uparrow),\fix_\a(\overline{f}))$ is in
  $\mathcal{L}_\a$. Then by Proposition~\ref{prop:uparrow_basics} we
  get $f(\fix_\a(\overline{f})^\uparrow)\leq
  \fix_\a(\overline{f})^\uparrow$.
\end{proof}

The above lemma guarantees that the sequence 
$f^n(\fix_\a(\overline{f})^\uparrow)$ is  decreasing. We can now
define an operator that, as we will show, is the fixpoint operator we
are looking for. 
\begin{defi}
  For every type $\a$ and $f\in\Kk_\a$ define 
  \begin{equation*}
    \Fix_\a (f)=
    \bigwedge_{n\in\mathbb{N}}(f^n(\fix_\a(\overline{f})^\uparrow))\ .
  \end{equation*}
\end{defi}
We show that $\Fix_\a$ is monotone.
\begin{lem}\label{lem:fixpoint_monotonicity}
  Given $f$ and $g$ in $\Kk_{\a\to \a}$, if $f\leq g$ then
  $\Fix_\a(f)\leq \Fix_\a(g)$.
\end{lem}

\begin{proof}
  By Lemma~\ref{lem:overline_ordering_lowering}, $f\leq g$ implies
  $\overline{f}\leq \overline{g}$, as $\fix_\a$ is monotone, we have
  $\fix_\a(\overline{f})\leq \fix_\a(\overline{g})$ and
  $\fix_\a(\overline{f})^\uparrow\leq \fix_\a(\overline{g})^\uparrow$
  by Proposition~\ref{prop:uparrow_basics}.  As $f\leq g$ we have
  $f^k(\fix_\a(\overline{f})^\uparrow)\leq
  g^k(\fix_\a(\overline{g})^\uparrow)$ for every $k$ in
  $\mathbb{N}$. Therefore
  $\bigwedge_{n\in\mathbb{N}}f^n(\fix_\a(\overline{f})^\uparrow)\leq
  \bigwedge_{n\in\mathbb{N}}g^n(\fix_\a(\overline{g})^\uparrow)$.
\end{proof}

The last step is to show that $\Fix_\a$ is actually in
$\Kk_{(\a\to\a)\to \a}$.

\begin{lem}\label{lem:fixpoint_overline}
  For every $\a$, $\Fix_\a$ is in $\Kk_\a$ and $(\Fix_\a,\fix_\a)$ is
  in $\mathcal{L}_{(\a\to\a)\to\a}$.
\end{lem}

\begin{proof}
  We know that $(f,\overline{f})$ in $\mathcal{L}_{\a\to\a}$. As we
  have seen in the proof of Lemma~\ref{lem:fixpoint_definition},
  $(f(\fix_\a(\overline{f})^\uparrow),\fix_\a(\overline{f}))$ is in
  $\mathcal{L}_\a$. Using repeatedly the defining properties of
  $\Ll_{\a\to\a}$, we obtain that for every $n\in\Nat$,
  $(f^n(\fix_\a(\overline{f})^\uparrow),\fix_\a(\overline{f}))$ is in
  $\mathcal{L}_\a$. But $f^{n}(\fix_\a(\overline{f})^\uparrow)$ is
  decreasing by Lemma~\ref{lem:fixpoint_definition}. Since $\Kk_\a$ is
  finite, we get $(\bigwedge_{n\in
    \mathbb{N}}f^n(\fix_\a(\overline{f})^\uparrow),\fix_\a(\overline{f}))$
  in $\Ll_\a$. We are done since $\bigwedge_{n\in
    \mathbb{N}}f^n(\fix_\a(\overline{f})^\uparrow)=\Fix_\a(f)$.
\end{proof}

\subsubsection{A model of the $\l Y$-calculus} We are ready to define the model we were looking for.

\begin{defi}\label{def:Kmodel}
  For a finite set $Q$ and its subset $Q_\W\incl Q$ consider a tuple
  $\Kk(Q,Q_\W,\r)=(\set{\Kk_\a}_{\a\in \Tt},\rho)$ where
  $\set{\Kk_\a}_{\a\in \Tt}$ is as in Definition~\ref{def:K} and
  $\rho$ is a valuation such that for every type $\a$: $\w^\a$ is
  interpreted as the greatest element of $\Kk_\a$,
  $Y^{(\a\to\a)\to\a}$ is interpreted as $\Fix_\a$, and $\W^\a$ is
  interpreted as $\bbot_\a$.
\end{defi}
Notice that, according to this definition, $\W^0$ is interpreted as
$(\bot,Q_\W)$. So the semantics of $\W$ and $\w$ are different in this
model. Recall that $\W$ is used to denote divergence, and  $\w$ is
used in the definition of the truncation operation from the semantics of
B\"ohm trees (cf. page~\pageref{eq:BT-semantics}).

We will show $\Kk(Q,Q_\W,\r)$ is indeed a model of the $\l Y$-calculus.
Since $\Kk_{\a\to\b}$ does not contain all the functions from $\Kk_\a$
to $\Kk_\b$ we must show that there are enough of them to form a model
of $\l Y$, the main problem being to show that $\sem{\l
  x. M}^\val_\Kk$ defines an element of $\Kk$.  For this, it is
sufficient to prove that constant functions and the combinators $S$
and $K$ exist in the model.

\begin{lem}\label{lem:combinatorial_completeness_of_Kk}
  For every sequence of types $\vec \a=\a_1\ldots\a_n$ and every types
  $\b$, $\g$ we have the following:
  \begin{itemize}
  \item For every constant $p\in \Kk_\b$ the constant function
    $f_p:\a_1\to\dots\to\a_n\to\b$ belongs to $\Kk$.
  \item For $i=1,\dots,n$, the projection
    $\pi_i:\a_1\to\dots\to\a_n\to\a_i$ belongs to $\Kk$.
  \item If $f:\vec\a\to(\b\to \g)$ and $g:\vec \a\to\b$ are in $\Kk$
    then $\lambda \vec p. f \vec p(g \vec p) :\vec\a\to\g$ is in $\Kk$.
  \end{itemize}
\end{lem}
\begin{proof}
  The first item of the lemma is given by
  Lemma~\ref{lem:constant_functions_are_in_Kk}, the second does not
  present more difficulty.  Finally, the third proceeds by a direct
  examination once we observe the following property of
  $\Kk(Q,Q_\W,\r)$.  Given two elements $f$ of
  $\mon{\Kk_{\a_1}\to\cdots \to \mon{\Kk_{\a_n}\to \Kk_{\b}}}$ and $g$
  of $\Dd_{\a_1\to\cdots\to\a_n\to\b}$, if for every $d_1$, \ldots,
  $d_n$ in $\Kk_{\a_1}$, \ldots, $\Kk_{\a_n}$,
 $   (f(d_1,\ldots, d_n),g(\overline{d_1},\ldots,\overline{d_n}))\in \Ll_{\b}$
  then $f$ is in $\Kk_{\a_1\to\cdots\to\a_n\to\b}$ and $(f,g)$ is in
  $\Ll_{\a_1\to\cdots\to\a_n\to\b}$.
  This observation follows directly from
  Proposition~\ref{prop:uparrow_basics} and the definition of the
  model. 
\end{proof}
The above lemma allows us to define the interpretation of terms in the
usual way:
\begin{itemize}
\item $\sem{Y^{(\b\to\b)\to\b}}^{\val}_{\Kk} =  \Fix_\b$
\item $\sem{a}_{\Kk}^{\val} =  \rho(a)$
\item $\sem{x^{\a}}^{\val}_{\Kk} = \val(x)$
\item $\sem{\w^\b}_{\Kk}^\val=\ttop_\b$
\item $\sem{\W^\b}_{\Kk}^\val=\bbot_\b$
\item $\sem{MN}_{\Kk}^{\val}=\sem{M}_{\Kk}^\val(\sem{N}^\val_{\Kk})$
\item $\sem{\l x^\a.M}_{\Kk}^\val(a) = \sem{M}_{\Kk}^{[\val[a/x]]}$,
  for every $a\in \Kk_\a$.
\end{itemize}
We need to check that for every valuation $\val$ and every term $M$ of type
$\a$, $\sem{M}^\val_{\Kk}$ is indeed in $\Kk_\a$. For this we take a
list of variables $x_1^{\a_1}$, \dots, $x_n^{\a_n}$ containing all
free varaibles of $M$, and we show that the function $\l p_1\dots
p_n. \sem{M}^{[p_1/x_1,\dots,p _n/x_n]}_\Kk$ is in
$\Kk_{\a_1\to\cdots\to\a_n\to\a}$.  The proof is a simple induction on
the structure of $M$. Lemma~\ref{lem:fixpoint_overline} and
Lemma~\ref{lem:combinatorial_completeness_of_Kk} ensure that this is
the case when $M=Y$. For the
other constants, $a$, $\w$ and $\W$, we use the fact that constant
functions are in the model. 
The remaining  cases are handled by 
Lemma~\ref{lem:combinatorial_completeness_of_Kk}: variable and
application clauses use $K$ and $S$ combinators respectively.

These observations allow us to conclude that $\Kk(Q,Q_\W,\r)$ is indeed a
model of the $\l Y$-calculus, that is:
  \begin{enumerate}
  \item for every term $M$ of type $\a$ and every valuation $\val$
    ranging of the free variables of $M$, $\sem{M}^\val_\Kk$ is in
    $\Kk_\a$,
  \item given two terms $M$ and $N$ of type $\a$, if $M=_{\b\d} N$,
    then for every valuation $\val$, $\sem{M}^\val_\Kk = \sem{N}^\val_\Kk$.
  \end{enumerate}

\begin{thm}\label{thm:Kmodel}
  For every finite set $Q$ and every set $Q_\W\incl Q$ the model
  $\Kk(Q,Q_\W,\r)$ as in Definition~\ref{def:Kmodel} is a model of the
  $\l Y$-calculus.
\end{thm}

Let us mention the following useful fact showing a
correspondence between the meanings of a term in $\Kk$ and in
$\Dd$. The proof is immediate since $\set{\Ll_\a}_{\a\in\Tt}$ is a
logical relation (cf~\cite{amadio98:_domain_lambd_calcul}).

\begin{lem}\label{lemma:K and D semantics}
  For every type $\a$ and closed term $M$ of type $\a$:
  $$(\sem{M}_\Kk,\sem{M}_\Dd)\in\Ll_\a\,.$$
\end{lem}


\subsection{Correctness and completeness of the
  model}\label{sec:corr-compl-model}

It remains to show that the model we have constructed is indeed
sufficient to recognize languages of TAC automata. For the rest of the
section we fix a tree signature $\S$ and a TAC automaton
\begin{equation*}
  \Aa=\struct{Q,\S,q^0\in Q,\d_1:Q\times \S_1\to
    \set{\ffalse,\ttrue},\d_2: Q\times \S_2 \to \Pp(Q^2)}\ .
\end{equation*}  

We take a model $\Kk$ based on $\Kk(Q,Q_\W,\r)$ as in
Definition~\ref{def:Kmodel}, where $Q_\W$ is the set of states $q$
such that $\d(q,\W)=\ttrue$. It remains to specify the meaning of
constants like $c:0$ or $a:0^2\to 0$ in $\S$\label{sec:corr-compl-model-1}:
\begin{align*}
  \r(c)=&(\top,\set{q : \d(q,c)=\ttrue})\\
  \r(a)(d_1,R_1)(d_2, R_2) =& (\top,R)\qquad \text{where
    $d_1,d_2\in\set{\bot,\top}$ and }\\
&  \quad R = \{q \in
Q\mid \d(q,a)\cap R_1\times R_2\not=\es\}\ .
\end{align*}

\begin{lem}
  For every $a$ in $\S$ of type $o^2\to o$: $\r(a)$ is in
  $\Kk_{o^2\to o}$ and $(\r(a),\top_{o^2\to o})$ is in
  $\mathcal{L}_{o^2\to o}$.
\end{lem}
\begin{proof}
  It is easy to see that $\r(a)$ is monotone. For the membership in
  $\Kk$ the witnessing function from $\Dd_{o^2\to o}$ is
  $\top_{0^2\to 0}$.
\end{proof}

Once we know that $\Kk$ is a model we can state some of its useful
properties. The first one tells what the meaning of unsolvable terms
is. The second indicates how unsolvability is taken into account in the
computation of a fixpoint.
\begin{prop}\label{prop:observing_convergence_in_Kk}
  Given a closed term $M$ of type $0$: $BT(M) = \W^0$ iff
  $\sem{M}_{\Kk} = (\bot,Q_\W)$.
\end{prop}

\begin{proof}
  If $\sem{M}_{\Kk}= (\bot,Q_\W)$ then Lemma~\ref{lemma:K and D
    semantics} gives us $\sem{M}_{\Dd} = \bot$. By
  Theorem~\ref{thm:D_and_convergence} this implies $BT(M) = \W^0$.

  If $BT(M) = \W^0$ then Theorem~\ref{thm:D_and_convergence} entails
  that $\sem{M}_{\Dd} = \bot$. By Lemma~\ref{lemma:K and D semantics}
  $(\sem{M}_{\Kk}, \bot)$ is in $\Ll_0$. But this is possible only if
  $\sem{M}_{\Kk} = (\bot,Q_\W)$.
\end{proof}

\begin{lem}\label{lemma:special-fix}
  Given a type $\b = \b_1\to\dots\to\b_l\to 0$, a sequence of types
  $\vec{\a} = \a_1,\dots,\a_k$, and a function $f\in \Kk_{\vec{\a}\to \b\to\b}$,
  consider the functions:
  \begin{equation*}
    h =\l p_1\dots p_k.\left( \fix_\b(\overline{ f (p_1)\dots
       ( p_k)})\right)\promote\qquad
    g = \l e_1\dots e_k.\fix_\b(\bar f (e_1)\dots (e_k))
  \end{equation*}
  that are respectively in
  $\mon{\Kk_{\a_1}\to\dots \to\mon{\Kk_{\a_{k}}\to\Kk_\b}}$ and in $\Dd_{\vec{\a}
    \to\b}$.  Then $h$ is in $\Kk_{\vec \a\to\b}$ and
  $(h,g)$ is in $\Ll_{\vec \a \to \b}$.  Moreover,  for every
  $p_1\in \Kk_{\a_1}$, \ldots, $p_k\in\Kk_{\a_k}$,
  $q_1\in\Kk_{\b_1}$,\ldots, $q_l\in \Kk_{\b_l}$ we have
  \begin{equation*}
   h (p_1,\dots ,p_k) (q_1,\dots , p_l)=
    \begin{cases}
      (\bot, Q_\W) & \text{if $g (\bar p_1,\dots, \bar p_k) (\bar q_1,\dots,\bar
        q_l)=\bot$}\\
      (\top ,Q) & \text {if $g(\bar p_1,\dots, \bar p_k)(\bar
        q_1,\dots,\bar q_l)=\top$}\ .
    \end{cases}
  \end{equation*}
\end{lem}

\begin{proof}
  To prove that $(h,g)$ is in $\Ll_{\vec \a\to \b}$, we resort to the
  remark we made in the proof of
  Lemma~\ref{lem:combinatorial_completeness_of_Kk}, so that it
  suffices to show that for every $p_1$, \ldots, $p_k$ respectively in
  $\Kk_{\a_1}$, \ldots, $\Kk_{\a_k}$, $(h(p_1,\ldots,p_k), g(\bar p_1,
  \dots ,\bar p_k))$ is in $\Ll_{\b}$.  We have that $h (p_1,\dots ,p_k) =
  \left( \fix_\b(\overline{ f (p_1,\dots,p_k)})\right)\promote$ that is
  in $\Kk_\b$, and then
  \begin{eqnarray*}
    \overline{h (p_1,\dots ,p_k)} &=& \overline{\left( \fix_\a(\overline{ f
          (p_1,\dots ,p_k)})\right)\promote}\\
    &=&\fix_\a(\overline{ f  (p_1,\dots ,p_k)})\\
    &=& \fix_\a(\bar f  (\bar p_1,\dots, \bar p_k)) \text{ by successive use of
    Lemma~\ref{lem:commutation_overline_application}}\\
  &=& g(\overline{p_1}, \dots, \bar p_k)\ .\\
\end{eqnarray*}
This shows that $(h,g)$ is in $\Ll_{\vec{\a}\to\b}$ and thus $h$ is in
$\Kk_{\vec{\a}\to\b}$.

So as to complete the proof of the lemma, we first prove the following
claim: for every for $r$ in $\Dd_{\g_1\to\dots\to\g_n\to 0}$, and
$q_1$, \dots, $q_n$ in $\Kk_{\g_1}$, \dots, $\Kk_{\g_n}$ we
have that:
\begin{itemize}
\item  $r^\uparrow(q_1,\dots, q_n) =(\bot,Q_\W)$ iff
  $(r(\overline{q_1},\dots,\overline{q_n}))^\uparrow = (\bot,Q_\W)$,
\item  $r^\uparrow(q_1,\dots, q_n) =(\top,Q)$ iff
  $(r(\overline{q_1},\dots,\overline{q_n}))^\uparrow = (\top,Q)$.
\end{itemize}

We first remark that, given $r$ in $\Dd_{\g\to\d}$, from the fourth
item of Proposition~\ref{prop:uparrow_basics}, we have that whenever
$(q,e)$ is in $\Ll_{\g}$, then $r^\uparrow(q) = (r(e))^\uparrow$, so
that in particular $r^\uparrow(q) = (r(\overline{q}))^\uparrow$.  A
simple induction shows then that, for $r$ in
$\Dd_{\g_1\to\dots\to\g_n\to\d}$, $$r^\uparrow(q_1,\dots, q_n) =
(r(\overline{q_1},\dots,\overline{q_n}))^\uparrow\ .$$ Therefore if $\d =
0$ and $r(\overline{q_1},\dots,\overline{q_n}) = \bot$, we have
$(r(\overline{q_1},\dots,\overline{q_n}))^\uparrow =
(\bot,Q_\W)$. Moreover, in case
$r(\overline{q_1},\dots,\overline{q_n}) = \top$, we have
$(r(\overline{q_1},\dots,\overline{q_n}))^\uparrow = (\top,Q)$. 

Now, the
lemma follows from choosing $r = g (\bar p_1,\dots,\bar p_k) $ and
remarking that we have $(g(\bar p_1,\dots ,\bar p_k))\promote =
h(p_1,\ldots,p_k)$.
\end{proof}

As in the case of GFP-models the semantics of a B\"ohm tree is defined
in terms of its truncations:
$\sem{BT(M)}_\Kk=\Land\set{\sem{BT(M)\dar_n}_\Kk \mid n\in \Nat}$.  The
subtle difference is that now $\W^0$ and $\w^0$ do not have the same
meaning. Nevertheless, the analog of Proposition~\ref{prop:semantics-of-BT}
still holds in $\Kk$.

\begin{thm}\label{thm:semantics-of-BT-in-K}
  For very closed term $M$ of type $0$: $\sem{M}_\Kk=\sem{BT(M)}_\Kk$.
\end{thm}
\begin{proof}
  First we show that $\sem{M}_\Kk\leq \sem{BT(M)}_\Kk$.  For this, we
  proceed with the classical finite approximation technique. We thus
  define a finite approximation of the B\"ohm tree. \emph{The Abstract
    B\"ohm tree up to depth $l$} of a term $M$, denoted $\ABT_l(M)$,
  will be a term obtained by reducing $M$ till it resembles $BT(M)$ up
  to depth $l$ as much as possible. We define it by induction:
\begin{itemize}
\item $\ABT_0(M)=M$;
\item $\ABT_{l+1}(M)$ is $M$ if $M$ does not have head normal form,\\
  otherwise it is a term $\l \vec{x}.N_0\ABT_l(N_1)\dots\ABT_l(N_k)$, where
  $\l \vec{x}.N_0N_1\dots N_k$ is the head normal form of $M$.
\end{itemize}

\noindent Since $\ABT_l(M)$ is obtained from $M$ by a sequence of
$\b\d$-reductions, $\sem{M}_\Kk=\sem{\ABT_l(M)}_\Kk$ for every $l$. We now
show that for every term $M$ and every $l$: 
\begin{equation*}
\sem{\ABT_l(M)}_\Kk\leq \sem{BT(M)\dar_l}_\Kk.  
\end{equation*}
Up to depth $l$, the two terms have the same structure as trees. We
will see that the meaning of every leaf in $\ABT_l(M)$ is not bigger
than the meaning of the corresponding leaf of $BT(M)\dar_l$. For
leaves of depth $l$ this is trivial since on the one hand we have a
term and on the other the constant $\w$. For other leaves, the terms
are either identical and thus have the same interpretation or on one
side we have a term without head normal form and on the other $\W^0$
and thus, according to
Proposition~\ref{prop:observing_convergence_in_Kk} also have the same
interpretation.

The desired inequality $\sem{M}_\Kk\leq\sem{BT(M)}_\Kk$ follows now
directly from the definition of the semantics of $\BT(M)$ since
$\sem{M}_\Kk=\sem{\ABT_l(M)}_\Kk\leq \sem{BT(M)\dar_l}_\Kk$ for every
$l\in \Nat$; and $\sem{BT(M)}_\Kk=\Land\set{\sem{BT(M)\dar_l}_\Kk \mid
l\in \Nat}$.

For the inequality in the other direction, we also use a classical
method that consists of working with finite unfoldings of the $Y$
combinators. 
Observe that if a term $M$ does not have $Y$ combinators, then
it is strongly normalizing and the theorem is trivial. So we need
be able to deal with $Y$ combinators in $M$. For this we introduce new
constants $c_{N}$ for every subterm $YN$ of $M$. The type of $c_{N}$
is $\vec{\a}\to\b$ if $\b$ is the type of $YN$ and $\vec\a = \a_1\dots
\a_k$ is the sequence of types of the sequence of free variables $\vec
x= x_1\dots x_k$ occurring in $YN$.  We let the semantics of a
constant $c_N$ be
\begin{equation*}
  \sem{c_N}_\Kk=\l\vec
p.\left(\fix_\b(\bar{\sem{N}_\Dd^{[\vec{p}/\vec{x}]}})\right)\promote\ .
\end{equation*}
First we need to check that indeed $\sem{c_N}_\Kk$ is in $\Kk$. For this
 we have prepared Lemma~\ref{lemma:special-fix}. Indeed
$\sem{c_N}_\Kk=\lambda p_1\dots p_k.\
\left(\fix_\b(\bar{f(p_1,\dots,p_k)})\right)^\uparrow$, for $f=\lambda \vec p.\
\sem{N}^{[\vec{p}/\vec{x}]}$. So $\sem{c_N}_\Kk$ is $h$ from
Lemma~\ref{lemma:special-fix} and $\sem{c_N}_\Dd=\bar{\sem{c_N}_\Kk}$
is $g$ from that lemma. The lemma additionally gives us that for every
$p_1,\dots,p_k$,$q_1,\dots,q_l$:
  \begin{equation}\label{eq:sem-of-fixpoint}
   \sem{c_N}_\Kk (p_1,\dots ,p_k) (q_1,\dots ,q_l)=
    \begin{cases}
      (\bot, Q_\W) & \text{if $\sem{c_N}_\Dd (\bar p_1,\dots,\bar p_k) (\bar q_1,\dots,\bar
        q_l)=\bot$}\\
      (\top ,Q) & \text {if $\sem{c_N}_\Dd (\bar p_1,\dots ,\bar p_k)(\bar q_1,\dots,\bar
        q_l)=\top$}\ .
    \end{cases}
  \end{equation}
\noindent We now define term $\iterate^n(N)$ for very $n\in \Nat$.
\begin{align*}
  \iterate^0(N)=& c_{N}\vec x\\
  \iterate^{n+1}(N)=&N(\iterate^n(N))\ .
\end{align*}
where $\vec x$ is the vector of variables free in $N$. Notice that
when replacing $c_N$ in $\iterate^n(N)$ by $\l \vec x. YN$ we obtain a
term that is $\b\d$-convertible to $YN$.

From the definition of the fixpoint operator in $\Kk$ and the fact
that $\Kk_\b$ is finite it follows that
$\sem{\l \vec x.\iterate^n(N)}=\sem{\l \vec x.YN}$ for some $n$.  Now
we can apply this identity to all fixpoint subterms in $M$ starting
from the innermost subterms. So the term $\expand^i(M)$ is obtained by
repeatedly replacing occurrences of subterms of the form $YN$ in $M$
by $\iterate^i(N)$ starting from the innermost occurrences. Now taking
$n$ so that for every $N$ occurring in $M$,
$\sem{\l \vec x.\iterate^n(N)}=\sem{\l \vec x.YN}$, we obtain
$\sem{M}_\Kk=\sem{\expand^n(M)}_\Kk$.

We come back to the proof. The missing inequality will be obtained
from
  \begin{equation*}
    \sem{M}_\Kk=\sem{\expand^n(M)}_\Kk=\sem{BT(\expand^n(M))}_\Kk\geq\sem{BT(M)}_\Kk\ .    
  \end{equation*}
  The first equality we have discussed above. The second is trivial
  since $\expand^n(M)$ does not have fixpoints. To finish the proof it
  remains to show 
  $\sem{BT(\expand^n(M))}_\Kk\geq\sem{BT(M)}_\Kk$.

  Let us denote $\BT(\expand^n(M))$ by $P$. So $P$ is a term of type
  $0$ in a normal form without occurrences of $Y$. For a term $K$ let
  $\tilde K$ stand for a term obtained from $K$ by simultaneously
  replacing $c_N$ by $\l \vec{x}.YN$. Because of
  Lemma~\ref{lem:fixpoint_definition}, we have $\sem{c_N}_\Kk\geq
  \sem{\l\vec x.YN}_\Kk$ which also implies that $\sem{K}_\Kk\geq
  \sem{\tilde K}_\Kk$. Moreover, as we have remarked above that
  replacing $c_N$ in $\iterate^n(N)$ by $\l \vec x. YN$ gives a term
  $\b\d$-convertible to $YN$, we have that $\tilde P$ is
  $\b\d$-convertible to $M$. It then follows that $BT(\tilde
  P)=BT(M)$. We need to show that $\sem{P}_\Kk\geq\sem{BT(\tilde
    P)}_\Kk$.

  Let us compare the trees $BT(P)$ and $BT(\tilde{P})$ by looking on
  every path starting from the root. The first difference appears when
  a node $v$ of $BT(P)$ is labeled with $c_N$ for some $N$. Say that
  the subterm of $P$ rooted in $v$ is $c_NK_1\dots K_i$. Then at the
  same position in $\BT(\tilde{P})$ we have the B\"ohm tree of the term $(\l
  \vec{x}. YN) \tilde K_1\dots\tilde K_i$. Observe that both terms are
  closed and of type $0$. This is because on the path from the root of
  $BT(P)$ to $v$ we have only seen constants of type $0\to 0\to 0$;
  similarly for $BT(\tilde{P})$.
  We will be done if we show that
  $\sem{c_NK_1\dots K_i}_\Kk\geq\sem{BT((\l \vec{x}. YN) \tilde
    K_1\dots\tilde K_i)}_\Kk$.

  We reason by cases. If $\sem{c_{N} K_1\dots K_i}_\Dd=\top$ then
  equation~\eqref{eq:sem-of-fixpoint} gives us $\sem{c_{N} K_1\dots
    K_i}_\Kk=(\top,Q)$.  So the desired inequality holds since
  $(\top,Q)$ is the greatest element of $\Kk_0$.

  If $\sem{c_{N} K_1\dots K_i}_\Dd=\bot$ then $\sem{c_{N} \tilde
    K_1\dots \tilde K_i}_\Dd=\bot$ since $\sem{K_i}_\Kk\geq
  \sem{\tilde K_i}_\Kk$. By equation~\eqref{eq:sem-of-fixpoint} we get
  $\sem{c_{N} \tilde K_1\dots \tilde K_i}_\Dd=(\bot,Q_\W)$. Since, by
  the definition of the fixpoint operator, $\sem{c_N}_\Kk\geq
  \sem{\lambda \vec x.\ YN}_\Kk$ we get $\sem{YN \tilde K_1\dots
    \tilde K_i}_\Kk=(\bot,Q_\W)$. But then
  Proposition~\ref{prop:observing_convergence_in_Kk} implies that $YN
  K_1\dots K_i$ is unsolvable. Thus $\sem{\BT((\l \vec x NY) \tilde
    K_1\dots \tilde K_i)}_\Kk=\sem{\W}_\Kk=(\bot, Q_\W)$.
\end{proof}



\begin{thm}
  Let $\Aa$ be an insightful TAC automaton with the set of states $Q$,
  initial state $q^0$, and $Q_\W$ the set of states from which $\Aa$
  accepts the constant $\W$. Let $\Kk=\Kk(Q,Q_\W)$ be a model as in
  Definition~\ref{def:Kmodel} where the constants have the
  interpretation $\r$ given page~\pageref{sec:corr-compl-model-1}. For every closed
  term $M$ of type $0$:
  \begin{equation*}
    BT(M)\in L(\Aa)\quad \text{iff}\quad \text{$q^0$ is in the second
      component of $\sem{M}_{\Kk}$}.
  \end{equation*}
\end{thm}
\begin{proof}
  The proof is very similar to the case of blind TAC automata
  (Proposition~\ref{prop:from_W_blind_TAC_to_GFP_models}). The
  difference here is that we rely on Theorem~\ref{thm:semantics-of-BT-in-K}
  for our model $\Kk$, moreover the constants $\w$ and $\W$ are
  handled separately. For completeness we spell out the argument in
  full, if only to see where these modifications intervene.

  For the left to right implication suppose that $\Aa$ accepts
  $\BT(M)$. Since, by Theorem~\ref{thm:semantics-of-BT-in-K},
  $\sem{M}=\sem{BT(M)}$ it is enough to show that $q^0$, that is the
  initial state of $\Aa$, is in the second component of
  $\sem{\BT(M)}$. For this we show that $q^0$ is in the second
  component of $\sem{\BT(M)\dar_l}$ for every $l\in M$.  

  The tree $\BT(M)$ is a ranked tree labeled with constants from the
  signature. The run of $\Aa$ is a function $r$ assigning to every node
  a state of $\Aa$. Recall that the tree $\BT(M)\dar_l$ is a prefix of this tree
  containing nodes up to depth $l$. Let us call it $t_l$. Every node
  $v$ in the domain of $t_l$ corresponds to a subterm of
  $\BT(M)\dar_l$ that we denote $M^l_v$.

  By induction on the height of $v$ we show that $r(v)$ appears in the
  second component of $\sem{M^l_v}$. This will show the left to right
  implication. If $v$ is a leaf at depth $l$ then
  $M^l_v$ is $\w^0$. We are done since $\sem{\w^0}=(\top,Q)$. If $v$ is
  a leaf of depth smaller than $l$ then $M^l_v$ is $\W^0$ or a constant
  $c$ of type $0$. In the latter case by definition of a run, we have
  $r(v)\in \set{q\mid \d(q,c)=\ttrue}$. We are done by the semantics of
  $c$ in the model. If $M^l_v$ is $\W^0$ then $\sem{M^l_v}=(\bot,Q_\W)$
  and $r(v)$ belongs to $Q_\W$ by definition of the run.  The last
  case is when $v$ is an internal node of the tree $t_l$. In this case
  $M^l_w=aM^l_{v1}M^l_{v2}$ where $a$ is the constant labeling $v$ in
  $t_l$. By the induction assumption we have that $r(vi)$ appears in the
  second component of $\sem{M^l_{vi}}$, and we are done by using the
  semantics of $a$.

  For the direction from right to left we suppose that $q^0$ is in the
  second component of $\sem{M}$. By
  Theorem~\ref{thm:semantics-of-BT-in-K}, $\sem{M}=\sem{BT(M)}$. We
  will construct a run of $\Aa$ on $BT(M)$.

  If $M$ does not have head normal form then $\sem{M}=(\bot,Q_\W)$ by
  Proposition~\ref{prop:observing_convergence_in_Kk}. In this case
  $\BT(M)$ is the tree consisting only of the root labeled
  $\W^0$. Hence $q^0\in Q_\W$ and we are done.

  Otherwise $\BT(M)$ has some letter $a$ in the root. In case it is a
  leaf, the conclusion is immediate. In case it is a binary symbol,
  $M=_{\b\d}a M_1M_2$ for some $M_1$, $M_2$. Now, as $q_0$ is in the
  second component of $\sem{M}$, by definition of $\sem{a}$, it must
  be the case that $q_1$ and $q_2$ are in the second components of
  $\sem{M_1}$ and $\sem{M_2}$, respectively. We put $r(1)=q_1$ and
  $r(2)=q_2$ and repeat the argument starting from the nodes $1$ and
  $2$ respectively. It is easy to see that this inductive procedure
  gives a, potentially infinite, run of $\Aa$. Hence
  $\BT(M)\in L(\Aa)$ as by construction the run of $\Aa$ is accepting.
  %
\end{proof}


\section{Reflection operation}\label{sec:reflection}
The idea behind the reflection operation is to transform a term into a
term that \emph{monitors} its computation: it is aware of the value
in the model of the original term at every moment of computation.
This monitoring simply amounts to adding an extra labelling to
constants that reflect those values. Formally, we express this by the
notion of a reflective B\"ohm tree defined below. The definition can
be made more general but we will be interested only in the case of
terms of type $0$. In this section we will show that reflective B\"ohm
trees can be generated by $\l Y$-terms.

As usual we suppose that we are working with a fixed tree signature
$\S$. We will also need a signature where constants are annotated with
elements of the model. If $\Ss=\struct{\set{\Ss_\a}_{\a\in\Tt},\r}$ is
a finitary model then the extended signature $\S^\Ss$ contains
constants $a^s$ where $a$ is a constant in $\S$ (either nullary or
binary) and $s\in \Ss_0$; so semantic annotations are possible
interpretations of terms of type $0$ in~$\Ss$.

\begin{defi}\label{df:Bohm tree respecting S}
  Let $\Ss$ be a finitary model, and $M$ a closed term of type $0$,
  $\rBT_\Ss(M)$, the \emph{reflective B\"ohm tree of $M$ with respect
    to $\Ss$}, is obtained in the following way:
\begin{itemize}
\item If $M\to^*_{\b\d} bN_1N_2$ for some constant $b:0\to 0\to 0$ 
  then $\rBT_\Ss(M)$ is a tree having the root 
  labelled  by $b^{\sem{bN_1N_2}_\Ss}$ and having
  $\rBT_\Ss(N_1)$ and $\rBT_\Ss(N_2)$ as subtrees.
\item If $M\to^*_{\beta \delta} c$ for some constant $c:0$ then
  $\rBT_\Ss(M)=c^{\sem{c}_\Ss}$. 
\item Otherwise, $M$ is unsolvable and $\rBT(M)=\W^0$.
\end{itemize}
\end{defi}
To see the intention behind this definition suppose that the model
$\Ss$ has the property: $\sem{N}_\Ss=\sem{\BT(N)}_\Ss$ for every term
$N$. In this case the superscript annotation of a node in
$\rBT_\Ss(M)$ is just the value of the subtree from this node. When,
moreover, the model $\Ss$ recognizes a given property then the
superscript determines if the subtree satisfies the property. For
example, GFP-models, as well as models $\Kk$ we have constructed in
the last section will behave this way. 

We will use terms to generate reflective B\"ohm trees.
\begin{defi}
  Let $\S$ be a tree signature, and let $\Ss$ be a finitary model. For $M$ 
  a closed term of type $0$ over the signature $\S$. We say that a
  term $M'$ over the signature $\S^\Ss$ is \emph{a reflection of $M$
    in $\Ss$} if $\BT(M')=\rBT(M)$.
\end{defi}
The objective of this section is to construct reflections of terms.
Since $\l Y$-terms can be translated to schemes and vice versa, the
construction is working for schemes too. (Translations between schemes
and $\l Y$-terms that do not increase the type order are presented
in~\cite{salvati12:_recur_schem_krivin_machin_collap_pushd_autom}).

Let us fix a tree signature $\S$ and a finitary model $\Ss$.  For the
construction of reflective terms we enrich the $\l Y$-calculus with some
syntactic sugar.  Consider a type $\a$. The set $\Ss_\a$ is finite for
every type $\a$; say $\Ss_\a=\set{d_1,\dots,d_k}$. We will introduce a
new atomic type $[\a]$ and constants $d_1,\dots,d_k$ of this type;
there will be no harm in using the same names for constants and elements
of the model. We do this for every type $\a$ and consider terms over
this extended type discipline. Notice that there are no
other closed normal terms than $d_1,\dots,d_k$ of type $[\a]$.

Given a term $M$ of type $\mbrack{\a}$ and $M_1$, \ldots $M_n$ which
are all terms of type $\b$, we introduce the construct
\begin{equation*}
\text{case}^\b M\set{d_i\to M_i}_{d_i\in \Ss_\a} 
\end{equation*}
which is a term of type $\b$ and which reduces to $M_i$ when $M=d_i$.
This construct is simple syntactic sugar since we may represent the
term $d_i$ of type $\mbrack{\a}$ with the $i^{\text{th}}$ projection
$\l x_1\ldots x_n.x_i$ by letting $\mbrack{\a} = 0^k\to 0$ then, when
$\b = \b_1\to\dots \to \b_n\to 0$, $\text{case}^\b$ can be defined as the
$\l$-term $$\l y_1^{\b_1}\dots y_n^{\b_n} d^{\mbrack{\a}} f^\b_1\dots
f^\b_k. d (f_1 y_1\dots y_n)\dots (f_k y_1\dots y_n) \ .$$ When $M$
represents $d_i$, \textit{i.e.} is equal to $\l x_1\ldots x_n.x_i$,
the term $$\l y_1^{\b_1}\dots y_n^{\b_n}. M (M_1 y_1\dots y_n)\ldots
(M_k y_1\dots y_n)$$ is $\b\eta$-convertible to $M_i$ which
represents well the semantic of the $\text{case}^\b$ construct.  In
the sequel, we shall omit the type annotation on the $\text{case}$
construct.

We define a transformation on types $\a^\bullet$ by induction on their
structure as follows:
\begin{eqnarray*}
  \a^\bullet&=&\a\text{ when }\a\text{ is atomic}\\
  (\a\to\b)^\bullet&=& \a^\bullet\to\mbrack{\a}\to\b^\bullet
\end{eqnarray*}
The type translation $(\cdot)^\bullet$ makes every function dependent
on the semantics of its argument.

The translation we are looking for will be an instance of a more
general translation  $\mbrack{M,\val}$ of a term $M$ of type $\a$
into a term of type $\a^\bullet$, where $\val$ is a valuation over
$\Ss$.

\begin{align*}
  \mbrack{\l x^\a.M,\val}=&\l x^{\a^\bullet}
 \l y^{\mbrack{\a}}.\\\
 & \quad\text{case }y^{\mbrack{\a}} \set{d \to \mbrack{M,\val[d/x^\a]}}_{d\in\Ss_\a}\\
\mbrack{MN,\val}=&
\mbrack{M,\val}\mbrack{N,\val}\sem{N}^\val\\
\mbrack{a,\val}=& \lambda
x^{0}_1 \l y^{\mbrack{0}}_1 \l x^{0}_2 \l y^{\mbrack{0}}_2.\\
&  \quad\text{case }y^{\mbrack{0}}_1\set{d_1 \to\text{case }
  y^{\mbrack{0}}_2 \set{d_2 \to a^{\r(a)d_1\,d_2}x_1x_2}_{d_2\in
  \Ss_0}}_{d_1\in \Ss_0}\\
&\text{when a is a binary constant}\\
\mbrack{a,\val}=& a^{\r(a)}\text{ when $a$ is a nullary constant}\\
\mbrack{x^\a,\val}=&x^{\a^\bullet}\\
\mbrack{Y^{(\a\to\a)\to\a} M,\val}=&Y^{(\a^\bullet\to\a^\bullet)\to\a^\bullet} (\l x^{\a^\bullet}.\mbrack{M,\val}
x^{\a^\bullet} \sem{Y M}^\val)\ .
\end{align*}
The transformation of the terms propagates semantic information. In
the case of $\l$-abstraction, the extra-semantic argument is checked
and in each branch the valuation is updated accordingly. In the case
of application, we need to give the extra semantic parameter, so we
simply give the interpretation of the argument in the model. For
constants, the term tests the value of each of the argument and then
sends the correctly annotated constant. For variables, we just need to
update their types. Finally for fixpoints, we type them with
$(\a^\bullet\to\a^\bullet)\to \a^\bullet$. When $M$ is the argument of
a fixpoint, the type of the term $\mbrack{M,\val}$, is
$(\a\to\a)^\bullet = \a^\bullet\to\mbrack{\a}\to\a^\bullet$.  We thus
take as an argument of $Y^{(\a^\bullet\to\a^\bullet)\to \a^\bullet}$
the term of type $\a^\bullet\to\a^\bullet$: $\l
x^{\a^\bullet}. \mbrack{M,\val}x^{\a^\bullet} \sem{Y M}^\val$ because
the semantics of the argument of $\mbrack{M,\val}$ is, by definition
of a fixpoint, the semantics of $YM$.

To prove correctness of this translation, we need two lemmas.

\begin{lem}\label{lem:substitution_lemma_for_reflexion}
  Given a term $M$ and a valuation $\val$, and the terms $N_1$,
  \ldots, $N_n$ we have the following
  identity:
  \begin{equation*}
    \mbrack{M\s,\val} =
    \mbrack{M,\val'}\s'\ ,
  \end{equation*}
  where $\s =  [N_1/x^{\a_1}_1,\ldots,N_n/x^{\a_n}_n]$ is a substitution,
  $\s'=[\mbrack{N_1,v}/x^{\a_1^\bullet}_1,\ldots,\mbrack{N_n,v}/x^{\a_n^\bullet}_n]$
  and $\val' = \val[\sem{N_1}^\val/x_1^{\a_1},\ldots,
  \sem{N_n}^\val/x_n^{\a_n}]$.
\end{lem}

\proof
  We proceed by induction on the structure of $M$. We will only show
  the case of $\lambda$-abstraction, the others being similar.



In case $M = \l x^\a.N$ (we assume that $x^\a$ is different from the
variables $x_i^{\a_i}$ used in the substitution), then $\mbrack{\l x^\a.M\s,\val} = \l x^{\a^\bullet}
y^{\mbrack{\a}}.\text{case } y^{\mbrack{\a}}\set{f \to
  M\s,\val[f/x^\a]}_{f\in \Ss_\a}$.  By induction we have that, for
every $f$ in $\Mm_\a$
$\mbrack{M\s,\val[f/x^\a]} = \mbrack{M,\val'[f/x^\a]}\s'$.  But,
\begin{eqnarray*}
  \mbrack{\l x^\a.M,\val'}\s' &=& (\l x^{\a^\bullet}
  y^{\mbrack{\a}}.\text{case } y^{\mbrack{\a}}\set{f\to
    \mbrack{M,\val'[f/x^\a]}}_{f\in \Ss_\a})\s'\\
  &=& \l x^{\a^\bullet}
  y^{\mbrack{\a}}.\text{case } y^{\mbrack{\a}}\set{f\to
    \mbrack{M,\val'[f/x^\a]}\s'}_{f\in \Ss_\a}\\
  &= &\l x^{\a^\bullet}
  y^{\mbrack{\a}}.\text{case } y^{\mbrack{\a}} \set{f\to
    \mbrack{M\s,\val[f/x^\a]}}_{f\in \Ss_\a}\\
  &=& \mbrack{\l x^\a.M\s,\val}\ .\rlap{\hbox to 214 pt{\hfill\qEd}}
\end{eqnarray*}


\noindent We can now show that the translation is compatible with head $\b\d$
reduction. 
\begin{lem}\label{lem:head_reduction_commutes_with_reflexion}
  If $M\to_{ h} M'$, then
  $\mbrack{M,\val}\to^{+}_{ h}\mbrack{M',\val}$.
\end{lem}

\proof
  We proceed by induction on the structure of $M$. We only treat the
  cases where $M$ is a redex, the other cases being trivial by
  induction.  We are left with two cases: $M = (\l x^\a.P) Q$ and $M =
  Y^{(\a\to\a)\to\a} P$.

  In case $M =  (\l x^\a.P) Q$,  we have that $M' = P[Q/x^\a]$, and
  using the Lemma~\ref{lem:substitution_lemma_for_reflexion} we have
  that $\mbrack{M',\val} =
  \mbrack{P,\val[\sem{Q}^\val/x^\a]}[\mbrack{Q,\val}/x^\a]$.
  But then we have
  \begin{eqnarray*}
    \mbrack{M,\val} &=& \mbrack{\l x^\a.P,\val} \mbrack{Q,\val}
    \sem{Q}^\val\\
    &=& (\l x^{\a^\bullet} y^{\mbrack{\a}}. \text{case }
    y^{\mbrack{\a}}\set{f\to \mbrack{P,\val[f/x^\a]}}_{f\in \Ss_\a})\mbrack{Q,\val}
    \sem{Q}^\val\\
    &\to_{ h}^+&
    \mbrack{P,\val[\sem{Q}^\val/x^\a]}[\mbrack{Q,\val}/x^\a]\\
    &=&\mbrack{M',\val}\ .
  \end{eqnarray*}

  In case $M= Y^{(\a\to\a)\to\a}P$, we have $M' = P M$ and:
  \begin{eqnarray*}
    \mbrack{M,\val}&=& Y^{(\a^\bullet\to\a^\bullet)\to\a^\bullet} (\l
    x^{\a^\bullet}. \mbrack{P,\val} x^{\a^\bullet} \sem{M}^\val)\\
    &\to_{ h}& (\l
    x^{\a^\bullet}. \mbrack{P,\val} x^{\a^\bullet} \sem{M}^\val)
    \mbrack{M,\val}\\
    &\to_{ h}& \mbrack{P,\val} \mbrack{M,\val}\sem{M}^\val\\
    &=&\mbrack{P M,\val}\\
    &=&\mbrack{M',\val}\ .\rlap{\hbox to 224 pt{\hfill\qEd}}
  \end{eqnarray*}



\begin{cor}\label{coro:solvability_head_reduction}
  Given a term $M$ of type $0$ and a valuation $\val$:
  \begin{equation*}
  M\to^{*}_{ h} aM_1 M_2\quad\text{iff}\quad
  \mbrack{M,\val}\to^{*}_{ h} a^{\sem{M}^\val}
  \mbrack{M_1,\val}\mbrack{M_2,\val}\ .    
  \end{equation*}
\end{cor}

\begin{proof}
  The direction from left to right is a simple consequence of
  Lemma~\ref{lem:head_reduction_commutes_with_reflexion}.  For the
  direction from right to left, we use the well-known fact
  (see~\cite{statman04}) that a $\l Y$-term has a head normal form iff
  it can be head-reduced to a head normal form.  Let us suppose that
  $\mbrack{M,\val}$ reduces to $a^{\sem{M}^\val} P_1 P_2$ in $k$ steps
  of head-reduction.  There are two cases. In case $M$ has no head
  normal form, then let $P$ be a term obtained from $M$ by $k+1$ steps
  of $\b\d$ reduction, in symbols $M{\to}^{k+1}_{ h} P$. By an
  iterative use of
  Lemma~\ref{lem:head_reduction_commutes_with_reflexion}, we must have
  $\mbrack{M,\val}{\to}^m_{ h}\mbrack{P,\val}$ with $k< m$. A
  contradiction since $P$ is not a head-normal form. The second case
  is when $M$ has a head-normal form. So after some number of steps of
  head $\b\d$-reduction we obtain $b N_1 N_2$. A simple use of
  Lemma~\ref{lem:head_reduction_commutes_with_reflexion} gives that $b
  = a$, $P_1 = \mbrack{N_1,\val}$ and $P_2=\mbrack{N_2,\val}$.
\end{proof}

A direct inductive argument using the above corollary gives us the
main result of this section.
\begin{thm}\label{thm:reflection}
  For every finitary model $\Ss$ and a closed term $M$ of type $0$:
  $$BT(\mbrack{M,\es})=\rBT_\Ss(M)\ .$$
\end{thm}





\paragraph{Remark:} If the divergence can be observed in the model
$\Ss$ (as it is the case for GFP models and for the model $\Kk$, cf.
Proposition~\ref{prop:observing_convergence_in_Kk}) then in the
translation above we could add the rule $\mbrack{M,\val} = \W$
whenever $\sem{M}^\val$ denotes a diverging term. We would obtain a
term which would always converge. A different construction for
achieving the same goal is proposed
in~\cite{DBLP:journals/corr/abs-1202-3498}.

\paragraph{Remark:} Even though the presented translation preserves
the structure of a term, it makes the term much bigger due to
the \emph{case} construction in the clause for $\l$-abstraction. The
blow-up is unavoidable due to complexity lower-bounds on the
model-checking problem. Nevertheless, one can try to limit the use of
the \emph{case} construct. We present below a slightly more efficient
translation that takes the value of the known arguments into account
and thus avoids the unnecessary use of the \emph{case} construction.  For
this, the translation is now parametrized also with a stack of values
from $\Ss$ so as to recall the values taken by the arguments.  For the
sake of simplicity, we also assume that the constants always have all
their arguments (this can be achieved by putting the $\l$-term in
$\eta$-long form). This translation is essentially obtained from the
previous one by techniques of constant propagation as used in partial
evaluation~\cite{jones93:_partial_evaluat_autom_progr_gener}.

\begin{eqnarray*}
  \mbrack{\l x^\a.M,\val,d::S} &=& \l x^{\a^\bullet}
  y^{\mbrack{\a}}.\mbrack{M,\val[d/x^\a],S}\\
  \mbrack{\l x^\a.M,\val,\e}&=&\l x^{\a^\bullet}
  y^{\mbrack{\a}}.\text{case }y^{\mbrack{\a}} \set{d \to
  \mbrack{M,\val[d/x^\a],\e}}_{d\in \Ss_\a}\\
\mbrack{MN,\val,S}&=&
\mbrack{M,\val,\sem{N}^\val::S}\mbrack{N,\val,\e}\sem{N}^\val\\
\mbrack{a,\val,d_1::d_2::\e}&=& \lambda x^0_1\l y^{\mbrack{0}}_1 \l
x^{0}_2\l y^{\mbrack{0}}_2.\
a^{\sem{a}d_1 d_2}x_1x_2\text{ when $a$ is a binary constant}\\  
\mbrack{a,\val}&=&a^{\r(a)}\text{ when $a$ is a nullary constant}\\
\mbrack{x^\a,\val,S}&=&x^{\a^\bullet}\\
\mbrack{Y M,\val,S}&=&Y \mbrack{M,\val,\sem{YM}^\val::S}\\
\end{eqnarray*}


\section{Conclusions}

We have considered the class of properties expressible by TAC
automata. These automata can talk about divergence as opposed to
$\W$-blind TAC automata that are usually considered in the
literature. We have given some example properties that require
TAC automata that are not $\W$-blind (cf.~page~\pageref{ex:properties}).
We have presented  the model-based approach to
model-checking problem for TAC automata.
While a
priori it is more difficult to construct a finitary model than to come
up with a decision procedure, in our opinion this additional effort is
justified. It allows, as we show here, to use the techniques of the
theory of the $\l$-calculus. It opens new ways of looking at the
algorithmics of the model-checking problem. Since typing in
intersection type systems~\cite{DBLP:conf/popl/Kobayashi09} and step
functions in models are in direct
correspondence~\cite{salvati12:_loader_urzyc_logic_relat}, the
model-based approach can also benefit from all the developments in
algorithms based on typing. Finally, this approach allows us to get
new constructions as demonstrated by our transformation of a scheme to
a scheme reflecting a given property. Observe that this transformation
is general and does not depend on our particular model.

As we have seen, the model-based approach is particularly
straightforward for $\W$-blind TAC automata. It uses standard
observations on models of the $\l Y$-calculus and
Proposition~\ref{prop:from_W_blind_TAC_to_GFP_models} with a simple
inductive proof. The model we propose for insightful automata may seem
involved; nevertheless, the construction is based on simple and
standard techniques. Moreover, this model implements an interesting
interaction between components. It succeeds in mixing a GFP model for
$\W$-blind automaton with the model $\Dd$ for detecting solvability.


The approach using models opens several new perspectives. One can
try to characterize which kinds of fixpoints correspond to which class
of automata conditions. More generally, models hint a possibility to
have an Eilenberg like variety theory for
lambda-terms~\cite{eilenberg}. This theory would cover infinite
regular words and trees too as they can be represented by $\l
Y$-terms. Finally, considering model-checking algorithms,
the model-based approach puts a focus on computing fixpoints in finite
partial orders. This means that a number of techniques, ranging from
under/over-approximations, to program optimization can be applied.



\bibliographystyle{alpha}
\bibliography{biblio}



\end{document}